\documentclass[journal]{IEEEtran}
\usepackage{cite}

\ifCLASSINFOpdf
  \usepackage[pdftex]{graphicx}
\else
\fi
\usepackage{amsmath,amsthm,amssymb}
\usepackage{mathtools}
\usepackage{cleveref}

\crefname{equation}{}{} 

\newtheorem{theorem}{Theorem}
\newtheorem{lemma}[theorem]{Lemma}
\newtheorem{corollary}[theorem]{Corollary}

\newtheorem{proposition}[theorem]{Proposition}
\newtheorem{definition}{Definition}

\theoremstyle{remark}
\newtheorem{remark}{Remark}

\newtheorem{example}{Example}

\newcommand{\bc}[1]{\left\{{#1}\right\}}
\newcommand{\br}[1]{\left({#1}\right)}
\newcommand{\bs}[1]{\left[{#1}\right]}
\newcommand{\abs}[1]{\left| {#1} \right|}

\newcommand{\E}[1]{\mathbb{E}\bs{{#1}}}

\newcommand{\cond}{\,|\,}
\newcommand{\bigcond}{\,\big|\,}

\DeclarePairedDelimiterX{\infdivx}[2]{(}{)}{  #1\;\delimsize\|\;#2}
\newcommand{\D}{D\infdivx}

\DeclarePairedDelimiterX{\cinfdivx}[3]{(}{)}{  #1\;\delimsize\|\;#2 \,\raisebox{-0.2ex}{\scalebox{1.5}{|}}\, #3}

\newcommand{\mc}{-\!\!\circ\!\!-}

\DeclareMathOperator{\SI}{SI}
\DeclareMathOperator{\II}{I}
\DeclareMathOperator{\HH}{H}

\DeclareMathOperator{\cII}{\mathcal{I}\,}

\DeclareMathOperator{\bias}{Bias}
\newcommand{\parent}{p}

\begin{document}
%
\title{Shared Information for a Markov Chain on a Tree\thanks{S. Bhattacharya and P. Narayan are with the Department of
Electrical and Computer Engineering and the Institute for Systems
Research, University of Maryland, College Park, MD 20742, USA.
E-mail: \{sagnikb, prakash\}@umd.edu. This work was supported by the U.S.
National Science Foundation under Grants CCF $1910497$ and
CCF $2310203$. A version of this paper was presented in part at the
2022 IEEE International Symposium on Information Theory~\cite{sb-pn-si-mct}.}}
%
%
%

\author{Sagnik Bhattacharya and Prakash Narayan}

\maketitle

\begin{abstract}
  Shared information is a measure of mutual dependence among multiple jointly distributed random variables with finite alphabets. For a Markov chain on a tree with a given joint distribution, we give a new proof of an explicit characterization of shared information. The Markov chain on a tree is shown to possess a global Markov property based on graph separation; this property plays a key role in our proofs. When the underlying joint distribution is not known, we exploit the special form of this characterization to provide a multiarmed bandit algorithm for estimating shared information, and analyze its error performance. 
\end{abstract}


%
\IEEEpeerreviewmaketitle

\section{Introduction}
%
%
%
%
\IEEEPARstart{L}{et} $X_1, \ldots, X_m$, $m\geq 2$ be random variables (rvs) with finite alphabets $\mathcal{X}_1, \ldots, \mathcal{X}_m$, respectively, and joint probability mass function (pmf) $P_{X_1 \cdots X_m}$. The \emph{shared information} $\SI(X_1, \ldots, X_m)$ of the rvs $X_1, \ldots, X_m$ is a measure of mutual dependence among them; and for $m=2$, $\SI(X_1, X_2)$ particularizes to mutual information $\II(X_1 \wedge X_2)$. Consider $m$ terminals, with terminal $i$ having privileged access to independent and identically distributed (i.i.d.) repetitions of $X_i$, $i = 1, \ldots, m$. Shared information $\SI(X_1, \ldots, X_m)$ has the operational meaning of being the largest rate of \emph{shared common randomness} that the $m$ terminals can generate in a distributed manner upon cooperating among themselves by means of interactive, publicly broadcast and noise-free communication\footnote{Our preferred nomenclature of shared information is justified by its operational meaning.}. Shared information measures the maximum rate of common randomness that is (nearly) independent of the open communication used to generate it.

The (Kullback-Leibler) divergence-based expression for $\SI(X_1, \ldots, X_m)$ was discovered in \cite[Example 4]{csiszar-narayan-secrecy-capacities}, where it was derived as an upper bound for a single-letter formula for the ``secret key capacity of a source model'' with $m$ terminals, a concept defined by the operational meaning above. The upper bound was shown to be tight for $m = 2$ and $3$. Subsequently, in a significant advance \cite{chan-tightness,chan-mutual-dependence,chan-hidden-flow}, tightness of the upper bound was established for arbitrary $m$, thereby imbuing $\SI(X_1, \ldots, X_m)$ with the operational significance of being the mentioned maximum rate of shared secret common randomness. The potential for shared information to serve as a natural measure of mutual dependence of $m \geq 2$ rvs, in the manner of mutual information for $m=2$ rvs, was suggested in \cite{narayan-isit}; see also \cite{tyagi-narayan-now}.

A comprehensive and consequential study of shared information, where it is termed ``multivariate mutual information'' \cite{chan-shared-information}, examines the role of secret key capacity as a measure of mutual dependence among multiple rvs and derives important properties including structural features of an underlying optimization along with connections to the theory of submodular functions. 

In addition to constituting secret key capacity for a multiterminal source model (\cite{csiszar-narayan-secrecy-capacities,chan-tightness,chan-mutual-dependence}), shared information also affords operational meaning for: maximal packing of edge-disjoint spanning trees in a multigraph (\cite{nitinawarat-secret-key,nitinawarat-perfect}; see also \cite{chan-linear-perfect,courtade-coded-cooperative}, \cite{chan-shared-information} for variant models); optimum querying exponent for resolving common randomness \cite{tyagi-how-many-queries}; strong converse for multiterminal secret key capacity \cite{tyagi-how-many-queries,tyagi-converses}; and also undirected network coding \cite{chan-hidden-flow}, data clustering \cite{chan-info-clustering}, among others. 

As argued in \cite{chan-shared-information}, shared information also possesses several attributes of measures of dependence among $m \geq 2$ rvs proposed earlier, including Watanabe's total correlation \cite{watanabe-tc} and Han's dual total correlation \cite{han-nonnegative} (both mentioned in \Cref{sec:preliminaries}). For $m=2$ rvs, measures of common information due to G\'acs-K\"orner \cite{gacs-common-information}, Wyner \cite{wyner-common-information} and Tyagi \cite{tyagi-common-information} have operational meanings; extensions to $m>2$ rvs merit further study (an exception \cite{liu-common-information} treats Wyner's common information).

For a given joint pmf $P_{X_1 \cdots X_m}$ of the  rvs $X_1, \ldots, X_m$, an explicit characterization of $\SI(X_1, \ldots, X_m)$ can be challenging (see \Cref{def:si} below); exact formulas are available for special cases (cf. e.g., \cite{csiszar-narayan-secrecy-capacities}, \cite{nitinawarat-secret-key}, \cite{chan-shared-information}). An efficient algorithm for calculating $\SI(X_1, \ldots, X_m)$ is given in \cite{chan-shared-information}. 

Our focus in this paper is on a Markov chain on a tree (MCT) \cite{Georgii+2011}. Tree-structured probabilistic graphical models are appealing owing to desirable statistical properties that enable, for instance, efficient algorithms for exact inference \cite{koller-friedman-pgm}, \cite{pearl-bayes}; decoding \cite{mackay-information-theory}, \cite{koller-friedman-pgm}; sampling \cite{feng-dynamic-sampling}; and structure learning \cite{chow-liu-trees}. An MCT can serve as a tractable tree-structured approximation to a given joint distribution arising in applications such as omniscience and secrecy generation \cite{csiszar-narayan-secrecy-capacities,chan-successive}, and signal clustering \cite{chan-clustering}. The mentioned tractability facilitates exact calculation of associated rate quantities. We take the tree structure of our model to be known; algorithms exist already for learning tree structure from data samples \cite{chow-liu-trees,chow-wagner}. We exploit the special form of $P_{X_1 \cdots X_m}$ in the setting of an MCT to obtain a simple characterization of shared information. When the joint pmf $P_{X_1 \cdots X_m}$ is not known but the tree structure is, the said characterization facilitates an estimation of shared information. 

In the setting of an MCT \cite{Georgii+2011}, our contributions are three-fold. First, we derive an explicit characterization of shared information for an MCT with a given joint pmf $P_{X_1 \cdots X_m}$ by means of a direct approach that exploits tree structure and Markovity of the pmf. A characterization of shared information had been sketched already in \cite{csiszar-narayan-secrecy-capacities}; our new proof does not seek recourse to a secret key interpretation of shared information, unlike in \cite{csiszar-narayan-secrecy-capacities}. Also, our proof differs in a material way from that in prior work \cite{chan-info-clustering} with a similar objective.

Second, we show an equivalence between the (weaker) original definition of an MCT \cite{Georgii+2011} and a (stronger) global  one based on separation in a graph \cite[Section 3.2.1]{lauritzen1996}. When $P_{X_1, \cdots, X_m}$ is assumed to be strictly positive, the two definitions are equivalent by the Hammersley-Clifford Theorem \cite[Theorem 3.9]{lauritzen1996}. We prove this equivalence even without said assumption, taking advantage of the underlying tree structure of the MCT; our proof method potentially is of independent interest.

Third, when $P_{X_1 \cdots X_m}$ is not known, with the mentioned characterization serving as a linchpin, we provide an approach for estimating shared information for an MCT. Formulated as a correlated bandits problem \cite{boda-prashanth-correlated-bandits}, this approach seeks to identify the best arm-pair across which mutual information is minimal. Using a uniform sampling of arms, redolent of sampling mechanisms in \cite{boda-narayan-universal-sampling-rate-distortion}, we provide an upper bound for the probability of estimation error and associated sample complexity. Our uniform sampling algorithm is similar to that in \cite{audibert-best-arm-identification}, \cite{boda-prashanth-correlated-bandits}; however, our modified analysis takes into account estimator bias, a feature that is not common in known bandit algorithms. Also, this approach can accommodate more refined bandit algorithms as also alternatives to the probability of error criterion such as regret \cite{cesa-bianchi_lugosi_2006}.

\Cref{sec:preliminaries} contains the preliminaries. Useful properties of an MCT are elucidated in \Cref{sec:mct-properties}. An explicit characterization of shared information for an MCT with a given $P_{X_1\cdots X_m}$ is provided in \Cref{sec:mact-shared-info}. \Cref{sec:estimating-si} describes our approach for estimating shared information when $P_{X_1\cdots X_m}$ is not known. \Cref{sec:concluding} contains closing remarks.

\section{Preliminaries}\label{sec:preliminaries}
    Let $X_1, \ldots, X_m$, $m \geq 2$, be rvs with finite alphabets $\mathcal{X}_1, \ldots, \mathcal{X}_m$, respectively, and joint pmf $P_{X_1 \cdots X_m}$. For $A \subseteq \mathcal{M} = \bc{1, \ldots, m}$, we write $X_A = (X_i, i \in  A)$ with alphabet $\mathcal{X}_A = \prod_{i \in A} \mathcal{X}_i$. Let $\pi = (\pi_1, \ldots, \pi_k)$ denote a $k$-partition of $\mathcal{M}$, $2 \leq k \leq m$. All logarithms and exponentiations are with respect to the base $2$, except $\ln$ and $\exp$ that are with respect to the base $e$.
    \begin{definition}[Shared information]\label{def:si}
        The \emph{shared information} of $X_1, \ldots, X_m$ is defined as
        \begin{align}
            &\SI(X_\mathcal{M}) \nonumber\\ 
            &= \min_{2 \leq k \leq m}\ \min_{\pi = (\pi_u, u =1, \cdots, k)}\ \frac{1}{k-1} \D{P_{X_\mathcal{M}}}{\prod_{u=1}^k P_{X_{\pi_u}}}\nonumber\\
            &= \min_{2 \leq k \leq m}\ \min_{\pi = (\pi_u, u =1, \cdots, k)}\ \frac{1}{k-1} \bs{\sum_{u=1}^k \HH(X_{\pi_u}) - \HH(X_{\mathcal{M}})}. \label{eq:si-def}
        \end{align}
    \end{definition}

    For a partition $\pi$ of $\mathcal{M}$ with $2 \leq \abs{\pi} \leq m$ atoms, it will be convenient to denote 
    \begin{align}
        \cII(\pi) = \frac{1}{\abs{\pi} - 1} \D{P_{X_\mathcal{M}}}{\prod_{u=1}^{\abs{\pi}} P_{X_{\pi_u}}} \label{eq:sub-si-eq}
    \end{align}
    so that $\SI(X_\mathcal{M}) = \min_{2 \leq \abs{\pi} \leq m} \cII(\pi)$.
    \begin{example}
        For $\mathcal{M} = \bc{1,2}$, we have 
        \begin{align*}
            \SI(X_1, X_2) = \text{mutual information } \II(X_1 \wedge X_2)
        \end{align*}
        and for $\mathcal{M} = \bc{1,2,3}$, it is checked readily that $\SI(X_1, X_2, X_3)$ is the minimum of $\II(X_1 \wedge X_2, X_3)$, $\II(X_2 \wedge X_1, X_3)$, $\II(X_3 \wedge X_1, X_2)$ and 
        \begin{align*}
            \frac{1}{2} \bs{\HH(X_1) + \HH(X_2) + \HH(X_3) - \HH(X_1, X_2, X_3)},
        \end{align*}
        and can be inferred from \cite[Examples 3,4]{csiszar-narayan-secrecy-capacities}.

        When $X_1, \ldots, X_m$ form a Markov chain $X_1 \mc \ldots \mc X_m$, it is seen that $\SI(X_\mathcal{M}) = \min_{1 \leq i \leq m-1} \II(X_i \wedge X_{i+1})$, the minimum mutual information between a pair of adjacent rvs in the chain \cite{csiszar-narayan-secrecy-capacities}.
    \end{example}

    Shared information possesses several properties befitting a measure of mutual dependence among multiple rvs. Clearly $\SI(X_\mathcal{M}) \geq 0$ and equality holds iff $P_{X_\mathcal{M}} = P_{X_A}P_{X_{A^c}}$ for some $A \subsetneq \mathcal{M}$; the latter follows from \cite[Theorem 5]{csiszar-narayan-secrecy-capacities} and \cite{chan-tightness}, \cite{chan-mutual-dependence}, \cite{chan-hidden-flow}. When $X_1, \ldots, X_m$ are bijections of each other, i.e., $\HH(X_i \cond X_j) = 0$, $1 \leq i \neq j \leq m$, then $\SI(X_\mathcal{M}) = \HH(X_1)$, as expected \cite{chan-shared-information}. 

    Next, the secret key capacity interpretation of $\SI(X_\mathcal{M})$ \cite{csiszar-narayan-secrecy-capacities,chan-tightness,chan-mutual-dependence,chan-hidden-flow,tyagi-narayan-now} implies that upon grouping the rvs $X_1, \ldots, X_m$ into disjoint teams represented by the atoms of any $k$-partition $\pi = (\pi_1, \ldots, \pi_k)$ of $\mathcal{M}$, $2 \leq k \leq m$, the resulting shared information of the teamed rvs $X_{\pi_1}, \ldots, X_{\pi_k}$ can be only larger, i.e.,
    \begin{align}
        \SI(X_{\pi_1}, \ldots, X_{\pi_k}) \geq \SI(X_1, \ldots, X_m). \label{eq:agglom-ineq}
    \end{align}
    Suppose that $\pi^* = (\pi_1^*, \ldots, \pi_l^*)$, $l \geq 2$, attains $\SI(X_\mathcal{M}) > 0$ (not necessarily uniquely) in \Cref{def:si}, i.e,
    \begin{align}
        \SI(X_\mathcal{M}) = \frac{1}{l-1} \D{P_{X_\mathcal{M}}}{\prod_{u=1}^l P_{X_{\pi_u^*}}}. \label{eq:pi-star-def}
    \end{align}
    A simple but useful observation based on \Cref{def:si}, \Cref{eq:agglom-ineq,eq:pi-star-def} is that upon agglomerating the rvs in each atom of an optimum partition $\pi^* = (\pi_1^*, \ldots, \pi_l^*)$, the resulting shared information $\SI(X_{\pi_1^*}, \ldots, X_{\pi_l^*})$ of the teams $X_{\pi_1^*}, \ldots, X_{\pi_k^*}$ equals the shared information $\SI(X_\mathcal{M})$ of the (unteamed) rvs $X_1, \ldots, X_m$, i.e, for these special teams \Cref{eq:agglom-ineq} holds with equality. This property has benefited information-clustering applications (cf. e.g., \cite{chan-shared-information}, \cite{chan-info-clustering}).

    Shared information satisfies the data processing inequality \cite{chan-shared-information}. For $X_\mathcal{M} = (X_1, \ldots, X_m)$, consider $X_\mathcal{M}' = (X_1', \ldots, X_m')$ where for a fixed $1 \leq j \leq m$, $X_i' = X_i$ for $i \in \mathcal{M} \setminus \bc{j}$ and $X_j'$ is obtained as the output of a stochastic matrix $W: \mathcal{X}_j \rightarrow \mathcal{X}_j$ with input $X_j$. Then, $\SI(X_{\mathcal{M}}') \leq \SI(X_\mathcal{M})$.

    It is worth comparing $\SI(X_\mathcal{M})$ with two well-known measures of correlation among $X_1, \ldots, X_m$, $m \geq 2$, of a similar vein. Watanabe's \emph{total correlation} \cite{watanabe-tc} is defined by 
    \begin{align}
        \mathcal{C}(X_\mathcal{M}) = \D{P_{X_\mathcal{M}}}{\prod_{i=1}^m P_{X_i}} = \sum_{i=1}^{m-1} \II(X_{i+1} \wedge X_1, \ldots, X_i) \label{eq:watanabe-tc}
    \end{align}
    and equals $(m-1) \cII(\pi)$ for the partition $\pi = (\bc{1}, \ldots, \bc{m})$ of $\mathcal{M}$ consisting of singleton atoms. By \Cref{eq:si-def,eq:watanabe-tc}, clearly 
    \begin{align}
        \SI(X_\mathcal{M}) \leq \frac{1}{m-1}\ \mathcal{C}(X_\mathcal{M}). \label{eq:si-tc-ineq}
    \end{align}
    Han's \emph{dual total correlation} \cite{han-nonnegative} is defined (equivalently) by 
    \begin{align}
        \mathcal{D}(X_\mathcal{M}) &= \sum_{i=1}^{m-1} \mathcal{D}iv{P_{X_i}}{P_{X_{i+1}\cdots X_m}}{P_{X_1\cdots X_{i-1}}} \nonumber \\
        &= \sum_{i=1}^m \HH(X_{\mathcal{M} \setminus \bc{i}}) - (m-1) \HH(X_\mathcal{M})\nonumber \\
        &= \HH(X_\mathcal{M}) - \sum_{i=1}^m \HH(X_i \cond X_{\mathcal{M} \setminus \bc{i}})\label{eq:1.9}\\
        &= \sum_{i=1}^{m-1} \II(X_i \wedge X_{i+1}, \ldots, X_m \cond X_1, \ldots, X_{i-1}),\nonumber
    \end{align} 
    (with conditioning vacuous for $i=1$) where the expression in \Cref{eq:1.9} is from \cite{Abdallah2010AMO}. By a straightforward calculation, these measures are seen to enjoy the sandwich 
    \begin{align}
        \frac{\mathcal{C}(X_\mathcal{M})}{m-1} \leq \mathcal{D}(X_\mathcal{M}) \leq (m-1)\ \mathcal{C}(X_\mathcal{M})\label{eq:1.10}
    \end{align}
    whereby we get from \Cref{eq:si-tc-ineq} and the first inequality in \Cref{eq:1.10} that 
    \begin{align}
        \SI(X_\mathcal{M}) \leq \frac{\mathcal{C}(X_\mathcal{M})}{m-1} \quad \text{and} \quad \SI(X_\mathcal{M}) \leq \mathcal{D}(X_\mathcal{M}).
    \end{align}
    This makes $\SI(X_\mathcal{M})$ a leaner measure of correlation than $\mathcal{C}(X_\mathcal{M})$ (upon setting aside the fixed constant $1/(m-1)$) or $\mathcal{D}(X_\mathcal{M})$. Significantly, the notion of an optimal partition in $\SI(X_\mathcal{M})$ in \Cref{eq:si-def} makes shared information an appealing measure for ``local'' as well as ``global'' dependencies among the rvs $X_1, \ldots, X_m$. 

    \begin{remark}
        When $\mathcal{M} = \bc{1,2}$, 
        \begin{align*}
            \SI(X_1, X_2) = \mathcal{C}(X_1, X_2) = \mathcal{D}(X_1, X_2) = \II(X_1 \wedge X_2).
        \end{align*}
    \end{remark}

    Our focus is on shared information for a Markov chain on a tree. 

    \begin{figure}[htbp]
        \centering
        \includegraphics[width=0.45\textwidth]{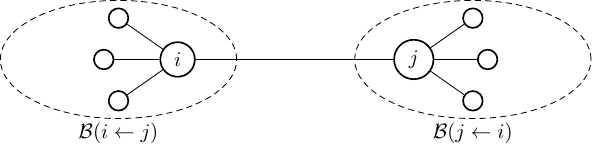}
        \caption{Notation for a Markov chain on a tree.}
        \label{diag:mct_diag}
    \end{figure}

    \begin{definition}[Markov Chain on a Tree]\label{def:mct-georgii}
        Let $\mathcal{G} = (\mathcal{M}, \mathcal{E})$ be a tree with vertex set $\mathcal{M} = \bc{1, \ldots, m}$, $m \geq 2$, i.e., a connected graph containing no circuits. For $(i,j)$ in the edge set $\mathcal{E}$, let $\mathcal{B}(i \leftarrow j)$ denote the set of all vertices connected with $j$ by a path containing the edge $(i,j)$. Note that $i \in \mathcal{B}(i \leftarrow j)$ but $j \notin \mathcal{B}(i \leftarrow j)$. See \Cref{diag:mct_diag}. The rvs $X_1, \ldots, X_m$ form a \emph{Markov Chain on a Tree} (MCT) $\mathcal{G}$ if for every $(i,j) \in \mathcal{E}$, the conditional pmf of $X_j$ given $X_{\mathcal{B}(i \leftarrow j)} = \bc{X_l: l \in \mathcal{B}(i \leftarrow j)}$ depends only on $X_i$. Specifically, $X_j$ is conditionally independent of $X_{\mathcal{B}(i \leftarrow j) \setminus \bc{i}}$ when conditioned on $X_i$. Thus, $P_{X_{\mathcal{M}}}$ is such that for each $(i,j) \in \mathcal{E}$, 
        \begin{align}\label{eq:mct}
            P_{X_j \cond X_{\mathcal{B}(i \leftarrow j)}} = P_{X_j \cond X_i},
        \end{align}
        or, equivalently,
        \begin{align}\label{eq:mct-mc-def}
            X_j \mc X_i \mc X_{\mathcal{B}(i \leftarrow j) \setminus \bc{i}}.
        \end{align}
        When $\mathcal{G}$ is a chain, an MCT reduces to a standard Markov chain.
    \end{definition}

    \begin{remark}
        With an abuse of terminology, we shall use $\mathcal{G} = (\mathcal{M}, \mathcal{E})$ to refer to a tree and also to the associated MCT.
    \end{remark}

    \begin{example}\label{eg:mct}
        Let $m = 2^l -1$ for some positive integer $l$. Consider a balanced binary tree with $l$ levels. Label the nodes progressively at each level and downwards, with the root node (at level $1$) being $1$ and the $2^{l-1}$ leaves (at level $l$) being $2^{l-1}, \ldots, 2^{l-1} + 2^{l-1} - 1 = 2^l -1 = m$. Let $X_1, Z_1, \ldots, Z_{m-1}$ be mutually independent rvs where $X_1 = \mathrm{Ber}(0.5)$ and $Z_i = \mathrm{Ber}(p_i)$ with $0 < p_i < 0.5$, $i=1, \ldots, m-1$. For $i=2, \ldots, m$, set $X_i = X_{\lfloor i/2 \rfloor } + Z_{i-1}$ where ``+'' denotes addition modulo $2$; note that $X_i$ is determined by $(X_1, Z_1, \ldots, Z_{i-1})$, and $P(X_i = 0) = P(X_i = 1) = 0.5$. 

        \begin{figure}[htbp]
            \centering
            \includegraphics[width=0.48\textwidth]{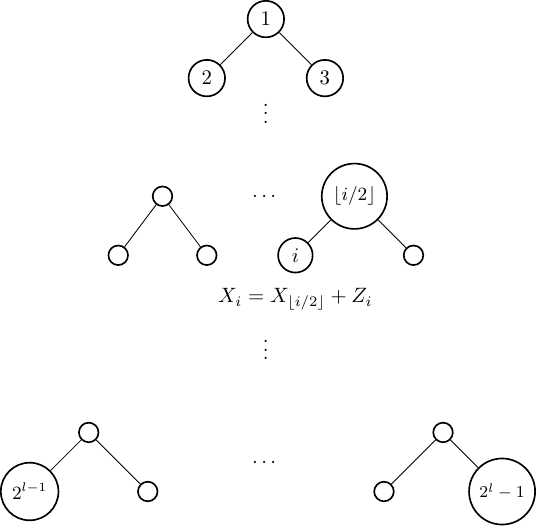}
            \caption{Example 2 of an MCT.}
            \label{diag:mct_example}
        \end{figure}
        
        Assign rv $X_i$ to vertex $i$, $i =1, \ldots, m$. See \Cref{diag:mct_example}. Then $X_1, \ldots, X_m$ form an MCT. Specifically, for any edge $(i,j)$ where $i$ is the parent of $j \geq 2$, we have
        \begin{align*}
            &P(X_j = x_j \cond X_{\mathcal{B}(i \leftarrow j)} = x_{\mathcal{B}(i \leftarrow j)})\nonumber\\
            &= P(X_j = x_j \cond X_i = x_i, X_{\mathcal{B}(i \leftarrow j) \setminus \bc{i}} = x_{\mathcal{B}(i \leftarrow j) \setminus \bc{i}})\\
            &= P(x_i + Z_{j-1} = x_j \cond X_i = x_i, X_{\mathcal{B}(i \leftarrow j) \setminus \bc{i}} = x_{\mathcal{B}(i \leftarrow j) \setminus \bc{i}})\\
            &= P(Z_{j-1} = x_j + x_i)\\
            &= P(X_j = x_j \cond X_i = x_i)
        \end{align*}
        where the last two inequalities are by the independence of $Z_{j-1}$ and $X_{\mathcal{B}(i \leftarrow j)}$, the latter rv being a function of $X_1, Z_{\bc{1, \ldots, m-1} \setminus \bc{j-1}}$.
    \end{example}

    \section{Properties of an MCT} \label{sec:mct-properties}
    We develop properties of an MCT that will play a role in characterizing shared information, and also are of independent interest. These include the concept of an agglomerated MCT, and notions of local and global Markov properties. 
    
    The main conclusion of this section is that the MCT as defined in \Cref{eq:mct-mc-def} has the global Markov property, which, in turn, implies \Cref{eq:mct-mc-def}; the proofs in this section, however, are based on \Cref{eq:mct-mc-def}. 
    
    We begin with agglomeration.
    \begin{lemma}\label{lem:mrf-mi}
        For the MCT $\mathcal{G} = (\mathcal{M}, \mathcal{E})$, for every $(i,j) \in \mathcal{E}$,
        \begin{align}
            \II(X_{\mathcal{B}(i \leftarrow j)} \wedge X_{\mathcal{B}(j \leftarrow i)}) = \II(X_i \wedge X_j),\label{eq:mrf-mi-2}
        \end{align}
        i.e.,
        \begin{align}
            X_{\mathcal{B}(i \leftarrow j) \setminus \bc{i}} \mc X_i \mc X_j \mc X_{\mathcal{B}(j \leftarrow i) \setminus \bc{j}}. \label{eq:mrf-mi-2b}
        \end{align}
    \end{lemma}
    \begin{proof}
        The assertion \Cref{eq:mrf-mi-2} is proved in \Cref{app:1}. Turning to \Cref{eq:mrf-mi-2b}, by the chain rule \Cref{eq:mrf-mi-2} is equivalent to 
        \begin{align*}
            X_{\mathcal{B}(i \leftarrow j) \setminus \bc{i}} \mc &X_i \mc X_j\\
            \text{and } X_{\mathcal{B}(i \leftarrow j)} \mc &X_j \mc X_{\mathcal{B}(j \leftarrow i) \setminus \bc{j}}
        \end{align*}
        which, in turn, are together tantamount to \Cref{eq:mrf-mi-2b}.
    \end{proof}

    \begin{definition}[Agglomerated Tree]\label{def:agglomerated-tree}
        Consider a $k$-partition $\pi = (\pi_1, \ldots, \pi_k)$ of $\mathcal{M}$, $2 \leq k \leq m-1$, where each $\pi_i$, $1 \leq i \leq k$, is connected. The tree $\mathcal{G}' = (\mathcal{M}', \mathcal{E}')$ with vertex set $\mathcal{M}' = \bc{\pi_1, \ldots, \pi_k}$ and edge set \begin{align*}
            \mathcal{E}' = \bc{(\pi_{i'}, \pi_{j'}): \exists i \in \pi_{i'}, j \in \pi_{j'} \text{ s.t. } (i,j) \in \mathcal{E}}
        \end{align*}
        is termed an agglomerated tree. Note that $\mathcal{G}' = (\mathcal{M}', \mathcal{E}')$ is a tree since $\mathcal{G}  = (\mathcal{M}, \mathcal{E})$ is a tree and each $\pi_i$, $1 \leq i \leq k$, is connected.
    \end{definition}
    \begin{lemma}\label{lem:agglom-mct}
        Consider the agglomerated tree $\mathcal{G}' = (\mathcal{M}', \mathcal{E}')$ in \Cref{def:agglomerated-tree}. If $X_1, \ldots, X_m$ form an MCT $\mathcal{G} = (\mathcal{M}, \mathcal{E})$, then $X_{\pi_1}, \ldots, X_{\pi_k}$ form an MCT $\mathcal{G}' = (\mathcal{M}', \mathcal{E}')$.
    \end{lemma}
    \begin{proof}
        By an obvious extension of \Cref{def:mct-georgii} to the agglomerated tree $\mathcal{G}' = (\mathcal{M}', \mathcal{E}')$, the lemma would follow upon showing that for every $(\pi_{i'}, \pi_{j'}) \in \mathcal{E}'$, it holds that 
        \begin{align}
            X_{\pi_{j'}} \mc X_{\pi_{i'}} \mc X_{\mathcal{B}(\pi_{i'} \leftarrow \pi_{j'}) \setminus \pi_{i'}}. \label{eq:agglom-mct-proof}
        \end{align}
        For any $(\pi_{i'}, \pi_{j'}) \in \mathcal{E}'$, there exist by \Cref{def:agglomerated-tree} $i \in \pi_{i'}$, $j \in \pi_{j'}$ with $(i,j) \in \mathcal{E}$. By \Cref{lem:mrf-mi},
        \begin{align}
            X_{\pi_{j'}} \mc X_i \mc X_{\mathcal{B}(i \leftarrow j) \setminus \bc{i}} \label{eq:agglom-mct-proof-2}
        \end{align}
        which, upon noting that $X_{\mathcal{B}(\pi_{i'} \leftarrow \pi_{j'}) \setminus \pi_{i'}} \subseteq X_{\mathcal{B}(i \leftarrow j) \setminus \bc{i}}$ and applying the chain rule of mutual information to \Cref{eq:agglom-mct-proof-2}, gives \Cref{eq:agglom-mct-proof}.
    \end{proof}
    Next, we turn to local and global Markovity of an MCT.
    \begin{definition}[Separation]
        For a tree $\mathcal{G} = (\mathcal{M}, \mathcal{E})$, let $A$, $B$ and $S$ be disjoint, nonempty subsets of $\mathcal{M}$. Then $S$ \emph{separates} $A$ and $B$ if for every $a \in A$ and $b \in B$, the path connecting $a$ and $b$ has at least one vertex $s = s(a,b)$ in $S$.
    \end{definition}
    The notion of maximally connected subset will be used below: $A' \subseteq A$ is a maximally connected subset of $A$ if $A'$ is connected and the addition to $A'$ of any vertex $u \in A \setminus A'$ renders $A' \cup \bc{u}$ to be disconnected. By convention, we shall take singleton elements to be connected. In general, any connected component of $A$ that is not maximally connected can be enlarged to absorb vertices outside it in $A$ that do not render the union to be disconnected.

    \begin{theorem}[Global Markov property]\label{th:g-g}
        For an MCT $\mathcal{G} = (\mathcal{M}, \mathcal{E})$, let $A$, $B$ and $S$ be disjoint, nonempty subsets of $\mathcal{M}$ such that $S$ separates $A$ and $B$. Then 
        \begin{align}
            X_A \mc X_S \mc X_B. \label{eq:g-g}
        \end{align}
        Conversely, if $X_1, \ldots, X_m$ are assigned to the vertices of a tree $\mathcal{G} = (\mathcal{M}, \mathcal{E})$ and satisfy \Cref{eq:g-g} for every such $A$, $B$, $S$, they form an MCT. 
    \end{theorem}
    \begin{remark}
        The Hammersley-Clifford Theorem \cite[Theorem 3.9]{lauritzen1996} implies the equivalence above for strictly positive joint pmfs, i.e., with $P(x_\mathcal{M}) > 0$ for every $x_\mathcal{M} \in \mathcal{X}_\mathcal{M}$. \Cref{th:g-g} shows that for an MCT, this equivalence holds also for a joint pmf $P_{X_\mathcal{M}}$ that is \emph{not} strictly positive.
    \end{remark}

    \begin{remark}
        The ``global'' Markov property \Cref{eq:g-g} \cite[Section 3.2.1]{lauritzen1996} clearly implies \Cref{eq:mct}, \Cref{eq:mct-mc-def} in \Cref{def:mct-georgii}. \Cref{th:g-g} asserts that for an MCT, \Cref{eq:g-g} is, in fact, equivalent to \Cref{eq:mct}, \Cref{eq:mct-mc-def}.
    \end{remark}

    Our proof of \Cref{th:g-g} in Appendix B relies on showing that an MCT has the ``local Markov property'' of \Cref{lem:mct-local-prop} below. The notion of \emph{neighborhood} is pertinent. \Cref{lem:mct-local-prop}, too, is proved in Appendix B.

    \begin{definition}[Neighborhood]
        For each $i \in \mathcal{M}$, its neighborhood is $\mathcal{N}(i) = \bc{j \in \mathcal{M}: (i,j) \in \mathcal{E}}$; note that $i \notin \mathcal{N}(i)$. Similarly, the neighborhood of $A \subsetneq \mathcal{M}$ is $\mathcal{N}(A) = \cup_{i \in A} \mathcal{N}(i) \setminus A$.
    \end{definition}
    
    \begin{lemma}[Local Markov property]\label{lem:mct-local-prop}
        Consider an MCT $\mathcal{G} = (\mathcal{M}, \mathcal{E})$. For each $i \in \mathcal{M}$, 
        \begin{align}
            X_i \mc X_{\mathcal{N}(i)} \mc X_{\mathcal{M} \setminus (\bc{i} \cup \mathcal{N}(i))}. \label{eq:mct-local-prop-1}
        \end{align}
        Furthermore, for every $A \subsetneq \mathcal{M}$ with no edge connecting any two vertices in it, 
        \begin{align}
            X_A \mc X_{\mathcal{N}(A)} \mc X_{\mathcal{M} \setminus (A \cup \mathcal{N}(A))}. \label{eq:mct-local-prop-2}
        \end{align}
    \end{lemma}
    \begin{remark}
        For $A \subsetneq \mathcal{M}$ as in \Cref{lem:mct-local-prop}, $\mathcal{N}(A) = \cup_{i \in A} \mathcal{N}(i)$.
    \end{remark}

    The relationships among the MCT definition \cref{eq:mct-mc-def}, and local and global Markov properties are summarized by the following lemma.
    \begin{lemma}
        It holds that\\ 
        (i) MCT and the global Markov property are equivalent;\\
        (ii) MCT implies the local Markov property;\\
        (iii) the local Markov property implies neither the global Markov property nor the MCT.
    \end{lemma}
    \begin{proof}
        (i), (ii) The proofs are contained in \Cref{th:g-g} and \Cref{lem:mct-local-prop}, respectively.\\
        (iii) The following example is used in \cite[Example 3.5]{lauritzen1996} to show that local Markovity does not imply global Markovity. Let $U$ and $Z$ be independent $\mathrm{Ber}(0.5)$ rvs, and let $W=U$, $Y=Z$ and $X = WY$. Consider the graph
        \begin{align*}
            U - W - X - Y - Z
        \end{align*}
        which has the local Markov property. Since 
        \begin{align*}
            \II(Z \wedge U,W \cond X) &= \II(Z \wedge U \cond UZ)\\
            &= \HH(Z \cond UZ)\\
            &= 0.75\,\HH(Z \cond UZ = 0)\\
            &= 0.75\, h(1/3),
        \end{align*}
        the claim in (iii) is true.
    \end{proof}

    \section{Shared Information for a Markov Chain on a Tree}\label{sec:mact-shared-info}
    We present a new proof of an explicit characterization of $\SI(X_\mathcal{M})$ for an MCT. The expression in \Cref{thm:si-mct} below was obtained first in \cite{csiszar-narayan-secrecy-capacities} relying on its secrecy capacity interpretation. Specifically, it was computed using a linear program for said capacity and seen to equal an upper bound corresponding to shared information (\!\!\!\cite[Examples 4, 7]{csiszar-narayan-secrecy-capacities}). The new approach below works directly with the definition of shared information in \Cref{def:si}. Also, it differs materially from the treatment in \cite[Section 4]{chan-info-clustering} for a model that appears to differ from ours.

    While the upper bound for $\SI(X_\mathcal{M})$ below is akin to that involving secret key capacity in \cite{csiszar-narayan-secrecy-capacities}, the proof of the lower bound uses an altogether new method based on the structure of a ``good'' partition $\pi$ in \Cref{def:si}.
    
    \begin{theorem}\label{thm:si-mct}
        Let $\mathcal{G} = (\mathcal{M}, \mathcal{E})$ be an MCT with pmf $P_{X_{\mathcal{M}}}$ in \Cref{eq:mct}. Then
        \begin{align}\label{eq:si}
            \SI(X_\mathcal{M}) = \min_{(i,j) \in \mathcal{E}} \II(X_i \wedge X_j).
        \end{align}
    \end{theorem}
    \begin{remark}\label{rem:min-edge}
        For later use, let $(\bar{i}, \bar{j})$ be the (not necessarily unique) minimizer in the right-side of \Cref{eq:si}
    \end{remark}
    \begin{example}
        For the MCT in \Cref{eg:mct}, $\SI(X_\mathcal{M}) = 1 - h(p^*)$, where $p^* = \max_{1 \leq i \leq m - 1} p_i < 0.5$. Thus, the $2$-partition obtained by cutting the (not necessarily unique) weakest correlating edge attains the minimum in \Cref{def:si}.
    \end{example}
    \begin{proof}
        As shown in \cite{csiszar-narayan-secrecy-capacities},
        \begin{align}\label{eq:si-ub}
            \SI(X_{\mathcal{M}}) \leq \min_{(i,j) \in \mathcal{E}} \II(X_i \wedge X_j)
        \end{align}
        and is seen as follows. For each $(i,j) \in \mathcal{E}$, consider a $2$-partition of $\mathcal{M}$, viz. $\pi = \pi((i,j)) = (\pi_1, \pi_2)$ where $\pi_1 = \mathcal{B}(i \leftarrow j)$, $\pi_2 = \mathcal{B}(j \leftarrow i)$. Then,
        \begin{align}
            \II(X_{\pi_1} \wedge X_{\pi_2}) &= \II(X_{\mathcal{B}(i \leftarrow j)} \wedge X_{\mathcal{B}(j \leftarrow i)}) \nonumber \\
            &= \II(X_i \wedge X_j), \text{ by \Cref{lem:mrf-mi}.} \label{eq:si-lb-a}
        \end{align}
        Hence, 
        \begin{align*}
            \SI(X_{\mathcal{M}}) \leq \II(X_{\pi_1} \wedge X_{\pi_2}) = \II(X_i \wedge X_j), \qquad (i,j) \in \mathcal{E} 
        \end{align*}
        leading to \Cref{eq:si-ub}.
        
        Next, we show that
        \begin{align}\label{eq:si-lb}
            \SI(X_{\mathcal{M}}) \geq \min_{(i,j) \in \mathcal{E}} \II(X_i \wedge X_j).
        \end{align}
        This is done in two steps. First, we show that for any $k$-partition $\pi$ of $\mathcal{M}$, $2 \leq k \leq m$, \emph{with (individually) connected atoms}, 
        \begin{align*}
            \cII(\pi) \geq \min_{(i,j) \in \mathcal{E}} \II(X_i \wedge X_j).
        \end{align*}
        Second, we argue that for any $k$-partition $\pi = (\pi_1, \ldots, \pi_k)$ containing disconnected atoms, there exists a $k'$-partition $\pi' = (\pi', \ldots, \pi'_{k'})$, possibly with $k' \neq k$, and with fewer disconnected atoms such that $\cII(\pi') \leq \cII(\pi)$.
        
        \textbf{Step 1:} Let $\pi = (\pi_1, \ldots, \pi_k)$, $k \geq 2$, be a $k$-partition with each atom being a connected set. By \Cref{lem:agglom-mct}, $X_{\pi_1}, \ldots, X_{\pi_k}$ form an agglomerated MCT $\mathcal{G}' = (\mathcal{M}', \mathcal{E}')$ as in \Cref{def:agglomerated-tree}. Furthermore, by \Cref{lem:mrf-mi}, if $\pi_u$ and $\pi_v$ in $\mathcal{M}'$ are connected by an edge $(\pi_u, \pi_v) \in \mathcal{E}'$, then there exist $i \in \pi_u$, $j \in \pi_v$, say, such that $(i,j) \in \mathcal{E}$, whence
        \begin{align}\label{eq:amct}
            \II(X_{\pi_u} \wedge X_{\pi_v}) = \II(X_i \wedge X_j).
        \end{align}
        Now, let $\pi_1, \ldots, \pi_k$ be an enumeration of the atoms, obtained from a breadth-first search \cite[Ch. 22]{clrs} run on the agglomerated tree with $\pi_1$ as the root vertex. Then, 
        \begin{align}
            \cII(\pi) &= \frac{1}{k-1} \D{P_{X_{\mathcal{M}}}}{\prod_{u=1}^k P_{X_{\pi_u}}}\nonumber\\
            &= \frac{1}{k-1} \bs{\sum_{u=1}^k \br{\HH(X_{\pi_u}) - \HH(X_{\pi_u} \cond X_{\pi_1}, \ldots, X_{\pi_{u-1}})}}\nonumber\\
            &= \frac{1}{k-1} \sum_{u=2}^{k} \II(X_{\pi_{u}} \wedge X_{\pi_1}, \ldots, X_{\pi_{u-1}})\nonumber\\
            &= \frac{1}{k-1} \sum_{u=2}^{k} \II(X_{\pi_{u}} \wedge X_{\mathrm{parent}(\pi_u)}) \label{eq:step-1-proof}\\
            &\geq \min_{(\pi_u,\pi_v) \in \mathcal{E}'} \II(X_{\pi_u} \wedge X_{\pi_v}) \nonumber\\
            &\geq \min_{(i,j) \in \mathcal{E}} \II(X_i \wedge X_j).\nonumber
        \end{align}
        By the breadth-first search algorithm \cite[Ch. 22]{clrs}, $\pi_1, \ldots, \pi_{u-1}$ are either at the same depth as $\pi_u$ or are above it (and include $\mathrm{parent}(\pi_u)$). This, combined with \Cref{th:g-g}, gives \Cref{eq:step-1-proof}. The last inequality is by \Cref{eq:amct}.

        \textbf{Step 2:} Consider first the case $k=2$. Take any $2$-partition $\pi = (\pi_1, \pi_2)$ with possibly disconnected atoms, where $\pi_1 = \cup_{\rho=1}^r C_\rho$ and $\pi_2 = \cup_{\sigma = 1}^s D_\sigma$ are unions of disjoint components. Since $\pi_1$ is connected to $\pi_2$, some $C_\rho$ and $D_\sigma$ must be connected by some edge $(i,j)$ in $\mathcal{E}$, so that   
        \begin{align*}
            &\cII(\pi) = \II(X_{\pi_1} \wedge X_{\pi_2}) \geq \II(X_{C_\rho} \wedge X_{D_\sigma})\\ 
            &\qquad\qquad\qquad\qquad\qquad\qquad\geq \II(X_i \wedge X_j) \geq \II(X_{\bar{i}} \wedge X_{\bar{j}})
        \end{align*}
        where the final lower bound, with $\bar{i}, \bar{j}$ as in \Cref{rem:min-edge}, is achieved by the $2$-partition with connected atoms $(\mathcal{B}(\bar{i} \leftarrow \bar{j}), \mathcal{B}(\bar{j} \leftarrow \bar{i}))$ as in \Cref{eq:si-lb-a}.

        Next, consider a $k$-partition $\pi = (\pi_1, \ldots, \pi_k)$, $k \geq 3$, and suppose that the atom $\pi_1$ is not connected. Without loss of generality, assume $\pi_1$ to be the (disjoint) union of maximally connected subsets $A_1, \ldots, A_t$, $t \geq 2$, of $\pi_1$ (which, at an extreme, can be the individual vertices constituting $\pi_1$).

        Take any $A_l$, say $A_l = A_{\bar{l}}$, and consider all its boundary edges, namely those edges for which one vertex is in $A_{\bar{l}}$ and the other outside it. As $A_{\bar{l}}$ is maximally connected in $\pi_1$, for each boundary edge the outside vertex cannot belong to $\pi_1$ and so must lie in $\mathcal{M} \setminus \pi_1$. Also, every such outside vertex associated with $A_{\bar{l}}$ must be the root of a subtree and, like $A_{\bar{l}}$, every $A_l$, $l \neq \bar{l}$, too, must be a subset of one such subtree linked to $A_{\bar{l}}$ -- owing to connectedness within $A_{\bar{l}}$. Furthermore, since $A_1, \ldots, A_t$ are connected, and only through the subtrees rooted in $\mathcal{M} \setminus \pi_1$, there must exist at least one $A_l$ such that all $A_{l'}$s, $l' \neq l$, are subsets of one subtree linked to $A_l$. In other words, denoting this $A_l$ as $A$, we note that $A$ has the property that 
        $$
            \pi_1 \setminus A = \bigcup_{\substack{l \in \bc{1, \cdots, t}:\\ A_l \neq A}} A_l
        $$ 
        is contained entirely in a subtree rooted at an outside vertex associated with $A$ and lying in $\mathcal{M} \setminus \pi_1$. Let this vertex be $j \in \mathcal{M} \setminus \pi_1$, and let $\pi_u \in \pi$ be the atom that contains $j$. Since vertex $j$ separates $A$ from $\pi_1 \setminus A$, so does $\pi_u$. By \Cref{th:g-g}, it follows that 
        \begin{align*}
            A \mc \pi_u \mc \pi_1 \setminus A
        \end{align*}
        whereby, using the data processing inequality,  
        \begin{align}
            \II(X_A \wedge X_{\pi_1 \setminus A}) \leq \II(X_{\pi_u} \wedge X_{\pi_1 \setminus A}) \leq \II(X_{\pi_u} \wedge X_{\pi_1}). \label{eq:6}
        \end{align}

        Next, consider the $(k-1)$-partition $\pi'$ and the $(k+1)$-partition $\pi''$ of $\mathcal{M}$, defined by 
        \begin{align}
            \pi' &= \br{\pi_1 \cup \pi_u, \bc{\pi_v}_{v \neq 1, v \neq u}},\label{eq:new-part-1}\\
            \pi'' &= \br{\pi_1 \setminus A, A, \pi_u, \bc{\pi_v}_{v \neq 1, v \neq u}}\label{eq:new-part-2}.
        \end{align}
        Then, 
        \begin{align*}
            \cII(\pi) &= \frac{1}{k-1} \Bigg[\HH(X_{\pi_1}) + \HH(X_{\pi_u}) \\
             &\qquad\qquad\qquad\qquad + \sum_{v \neq 1, v \neq u} \HH(X_{\pi_v}) - \HH(X_{\mathcal{M}})\Bigg],\\
            \cII(\pi') &= \frac{1}{k-2} \Bigg[\HH(X_{\pi_1 \cup \pi_u}) + \sum_{v \neq 1, v \neq u} \HH(X_{\pi_v}) \\
            &\qquad\qquad\qquad\qquad\qquad\qquad\qquad\qquad - \HH(X_{\mathcal{M}})\Bigg],\\
            \cII(\pi'') &= \frac{1}{k} \Bigg[\HH(X_{\pi_1 \setminus A}) + \HH(X_A) + \HH(X_{\pi_u}) \\
            &\qquad\qquad\qquad\qquad+ \sum_{v \neq 1, v \neq u} \HH(X_{\pi_v}) - \HH(X_{\mathcal{M}})\Bigg].
        \end{align*}
        We claim that 
        \begin{align}
            \cII(\pi) \geq \min\bc{\cII(\pi'), \cII(\pi'')} \label{eq:main-proof-claim}.
        \end{align}
        Referring to \Cref{eq:new-part-1,eq:new-part-2}, we can infer from the claim \Cref{eq:main-proof-claim} that for a given $k$-partition $\pi$ with a disconnected atom $\pi_1$ as above, merging a disconnected atom with another atom (as in \Cref{eq:new-part-1}) or breaking it to create a connected atom (as in \Cref{eq:new-part-2}), lead to partitions $\pi'$ or $\pi''$, of which at least one has a lower $\cII$-value than $\pi$. This argument is repeated until a final partition with connected atoms is reached that has the following form: considering the set of all maximally connected components of the atoms of $\pi = (\pi_1, \ldots, \pi_k)$, the final partition will consist of \emph{connected} unions of such components. (A connected $\pi_i$ already constitutes such a component.)

        It remains to show \Cref{eq:main-proof-claim}. Suppose \Cref{eq:main-proof-claim} were not true, i.e., 
        \begin{align*}
            \cII(\pi) < \min\bc{\cII(\pi'), \cII(\pi'')}.
        \end{align*}
        Then, 
        \begin{align}
            \cII(\pi) < \cII(\pi') &\Leftrightarrow (k-2) \cII(\pi) < (k-2) \cII(\pi')\nonumber\\
            &\Leftrightarrow \II(X_{\pi_u} \wedge X_{\pi_1}) < \cII(\pi), \label{eq:10}
        \end{align}
        and similarly,
        \begin{align}
            \cII(\pi) < \cII(\pi'') &\Leftrightarrow k \cII(\pi) < k \cII(\pi'')\nonumber\\
            &\Leftrightarrow \cII(\pi) < \II(X_{\pi_1 \setminus A} \wedge X_A) \label{eq:11}
        \end{align}
        where the second equivalences in \Cref{eq:10,eq:11} are obtained by straightforward manipulation. By \Cref{eq:10,eq:11},
        \begin{align*}
            \II(X_{\pi_u} \wedge X_{\pi_1}) < \II(X_{\pi_1 \setminus A} \wedge X_A)
        \end{align*}
        which contradicts \Cref{eq:6}. Hence, \Cref{eq:main-proof-claim} is true. 
    \end{proof}

    \section{Estimating $\SI(X_\mathcal{M})$ for an MCT}\label{sec:estimating-si}
    We consider the estimation of $\SI(X_\mathcal{M})$ when the pmf $P_{X_\mathcal{M}}$ of $X_\mathcal{M} = (X_1, \ldots, X_m)$ is unknown to an ``agent'' who, however, is assumed to know the tree $\mathcal{G} = (\mathcal{M}, \mathcal{E})$. We assume further in this section that $\mathcal{X}_1 = \cdots = \mathcal{X}_m = \mathcal{X}$, say, and also that the minimizing edge $(\bar{i}, \bar{j})$ on the right side of \Cref{eq:si} is unique. By \Cref{thm:si-mct}, $\SI(X_\mathcal{M})$ equals the minimum mutual information across an edge in the tree $\mathcal{G}$. Treating the determination of this edge as a correlated bandits problem of best arm pair identification, we provide an algorithm to home in on it, and analyze its error performance and associated sample complexity. \emph{The estimate of shared information is taken to be the mutual information across the best arm-pair thus identified.} Our estimation procedure is motivated by the form of $\SI(X_\mathcal{M})$ in \Cref{thm:si-mct}.
    
    \subsection{Preliminaries}
        As stated, estimation of $\SI(X_\mathcal{M})$ for an MCT will entail estimating $\II(X_i \wedge X_j)$, $(i,j) \in \mathcal{E}$. We first present pertinent tools that will be used to this end.

        Let $(X_t, Y_t)_{t=1}^n$ be $n \geq 1$ independent and identically distributed (i.i.d.) repetitions of rv $(X,Y)$ with (unknown) pmf $P_{XY}$ of \emph{assumed full support} on $\mathcal{X} \times \mathcal{Y}$, where $\mathcal{X}$ and $\mathcal{Y}$ are finite sets. For $(\mathbf{x}, \mathbf{y})$ in $\mathcal{X}^n \times \mathcal{Y}^n$, let $Q_{\mathbf{x}\mathbf{y}}^{(n)}$ represent its joint type on $\mathcal{X} \times \mathcal{Y}$ (cf. \cite[Ch. 2, 3]{csiszar2011information}). Also, let $Q_\mathbf{x}^{(n)}$ (resp. $Q_\mathbf{y}^{(n)}$) represent the (marginal) type of $\mathbf{x}$ (resp. $\mathbf{y}$).

        A well-known estimator for $\II(X \wedge Y) = \II_{P_{XY}}(X \wedge Y)$ on the basis of $(\mathbf{x}, \mathbf{y})$ in $\mathcal{X}^n \times \mathcal{Y}^n$ is the \emph{empirical mutual information} (EMI) estimator ${\II}_{\mathsf{EMI}}^{(n)}$, based on EMI \cite{goppa-mi}, \cite[Ch. 3]{csiszar2011information},  defined by 
        \begin{align}
            {\II}_{\mathsf{EMI}}^{(n)}(\mathbf{x} \wedge \mathbf{y}) = \HH(Q^{(n)}_{\mathbf{x}}) + \HH(Q^{(n)}_{\mathbf{y}}) - \HH(Q^{(n)}_{\mathbf{x}\mathbf{y}}). \label{eq:1000}
        \end{align}
        Throughout this section, $(\mathbf{X},\mathbf{Y})$ will represent $n$ i.i.d. repetitions of the rv $(X,Y)$. 
        \begin{lemma}[Bias of EMI estimator]\label{lem:mut-inf-bias}
            The bias 
            \begin{align*}
                \bias({\II}_{\mathsf{EMI}}^{(n)}(\mathbf{X} \wedge \mathbf{Y})) \triangleq \mathbb{E}_{P_{XY}}\bs{{\II}_{\mathsf{EMI}}^{(n)}(\mathbf{X} \wedge \mathbf{Y})} - \II(X \wedge Y)
            \end{align*}
            satisfies
            \begin{align*}
                &- \log \br{1 + \frac{\abs{\mathcal{X}}-1}{n}}\br{1 + \frac{\abs{\mathcal{Y}}-1}{n}} \\
                &\qquad\qquad\leq \bias({\II}_{\mathsf{EMI}}^{(n)}(\mathbf{X} \wedge \mathbf{Y})) \leq \log \br{1 + \frac{\abs{\mathcal{X}}\abs{\mathcal{Y}}-1}{n}}.
            \end{align*}
        \end{lemma}
        \begin{proof}
            The proof follows immediately from \cite[Proposition 1]{paninski}.
        \end{proof}

        A concentration bound for the estimator ${\II}_{\mathsf{EMI}}^{(n)}$ in \Cref{eq:1000} using techniques from \cite{kontoyiannis}, is given by
        \begin{lemma}\label{lem:emp-mi-conc}
            Given $\epsilon > 0$ and for every $n \geq 1$, 
            \begin{align*}
                &P_{XY} \br{{\II}_{\mathsf{EMI}}^{(n)}(\mathbf{X} \wedge \mathbf{Y}) - \mathbb{E}_{P_{XY}}\bs{{\II}_{\mathsf{EMI}}^{(n)}(\mathbf{X} \wedge \mathbf{Y})} \geq \epsilon} \\
                &\qquad\qquad\qquad\qquad\qquad\qquad\qquad\leq \exp \br{- \frac{2n \epsilon^2}{36\log^2 n}}.
            \end{align*}
            The same bound applies upon replacing ${\II}_{\mathsf{EMI}}^{(n)}(\mathbf{X} \wedge \mathbf{Y})$ by $-{\II}_{\mathsf{EMI}}^{(n)}(\mathbf{X} \wedge \mathbf{Y})$ above.
        \end{lemma}
        \begin{proof}
            The empirical mutual information ${\II}_{\mathsf{EMI}}^{(n)}: \mathcal{X}^n \times \mathcal{Y}^n \rightarrow \mathbb{R}^+ \cup \bc{0}$ satisfies the bounded differences property, namely
            \begin{align}
                &\max_{\substack{(\mathbf{x}, \mathbf{y}) \in \mathcal{X}^n \times \mathcal{Y}^n \\ (x_i', y_i') \in \mathcal{X} \times \mathcal{Y}}} \bigg|{\II}_{\mathsf{EMI}}^{(n)}(\mathbf{x} \wedge \mathbf{y}) \nonumber\\
                &\qquad- {\II}_{\mathsf{EMI}}^{(n)}((x_1^{i-1}, x_i', x_{i+1}^n) \wedge (y_1^{i-1}, y_i', y_{i+1}^n))\bigg| \leq \frac{6\log n}{n} \label{eq:emp-mi-bd}
            \end{align}
            for $1 \leq i \leq n$, where for $l < k$, $x_l^k = (x_l, x_{l+1}, \ldots, x_k)$. 
            
            To see this, we note that changing 
            \begin{align*}
                &(\mathbf{x}, \mathbf{y}) = ((x_1, \ldots, x_n), (y_1, \ldots, y_n)) \\
                &\qquad\qquad\qquad\rightarrow ((x_1^{i-1}, x_i', x_{i+1}^n), (y_1^{i-1}, y_i', y_{i+1}^n))
            \end{align*}    
            amounts to changing at most two components in the joint type $Q_{\mathbf{x}\mathbf{y}}^{(n)}$ and marginal types $Q_\mathbf{x}^{(n)}$ and $Q_\mathbf{y}^{(n)}$; in each of these three cases, the probability of one symbol or one pair of symbols decreases by $1/n$ and that of another increases by $1/n$. The difference between the corresponding empirical entropies is given in each case by the sum of two terms. For instance, one such term for the joint empirical entropy is given by 
            \begin{align*}
                &\bigg|\ Q^{(n)}_{\mathbf{x}\mathbf{y}}(x_i, y_i) \log Q^{(n)}_{\mathbf{x}\mathbf{y}}(x_i, y_i)\nonumber \\
                &\qquad\qquad- \br{Q^{(n)}_{\mathbf{x}\mathbf{y}}(x_i, y_i) - \frac{1}{n}} \log \br{Q^{(n)}_{\mathbf{x}\mathbf{y}}(x_i, y_i) - \frac{1}{n}}\bigg|.
            \end{align*}
            Each of these terms is $\leq \log n/n$, using the inequality \cite{kontoyiannis}
            \begin{align*}
                \abs{\frac{j+1}{n} \log \frac{j+1}{n} - \frac{j}{n} \log \frac{j}{n}} \leq \frac{\log n}{n}, \qquad 0 \leq j < n.
            \end{align*}
            The bound in \Cref{eq:emp-mi-bd} is obtained upon applying the triangle inequality twice in each of the three mentioned cases. The claim of the lemma then follows by a standard application of McDiarmid's Bounded Differences Inequality \cite[Theorem 2.9.1]{vershynin_2018}.
        \end{proof}
        Since we seek to identify the edge with the smallest mutual information across it, we next present a technical lemma that bounds above the probability that the estimates of the mutual information between two pairs of rvs are in the wrong order. Our proof uses \Cref{lem:emp-mi-conc}. Let $(X,Y)$ and $(X',Y')$ be two pairs of rvs with pmfs $P_{XY}$ and $P_{X'Y'}$, respectively, on the (common) alphabet $\mathcal{X} \times \mathcal{Y}$, such that $\II(X \wedge Y) < \II(X' \wedge Y')$. Let 
        \begin{align}
            \Delta = \II(X' \wedge Y') - \II(X \wedge Y) > 0. \label{eq:delta-def}
        \end{align}
        By \Cref{lem:mut-inf-bias}, ${\II}_{\mathsf{EMI}}^{(n)}$ is asymptotically unbiased and, in particular, we can make $\bias(\II_\mathsf{EMI}^{(n)}(\mathbf{X} \wedge \mathbf{Y})), \bias(\II_\mathsf{EMI}^{(n)}(\mathbf{X}' \wedge \mathbf{Y}')) < \Delta/2$ by choosing $n$ large enough, for instance, 
        \begin{align}
            n > \max\bc{\frac{\abs{\mathcal{X}}^2 - 1}{2^{\Delta/2} - 1}, \frac{\abs{\mathcal{X}} - 1}{2^{\Delta/4} - 1}}. \label{eq:lower-bound-on-n}
        \end{align}
        The upper bound on the probability of ordering error depends on the bias of ${\II}_{\mathsf{EMI}}^{(n)}$ and decreases with decreasing bias.
        \begin{lemma}\label{lem:estim}
            With $(X,Y)$, $(X',Y')$ and $n$ as in \Cref{eq:lower-bound-on-n},
            \begin{align*}
                &P\br{\II^{(n)}_{\mathsf{EMI}}(\mathbf{X} \wedge \mathbf{Y}) \geq \II^{(n)}_{\mathsf{EMI}} (\mathbf{X}' \wedge \mathbf{Y}')} \nonumber \\
                &\leq 2 \max\left\{\exp \br{- \frac{2n \br{{\Delta}/{2} - \bias\br{\II^{(n)}_{\mathsf{EMI}} (\mathbf{X} \wedge \mathbf{Y})}}^2}{36 \log^2 n}}, \right.\nonumber\\
                &\qquad\qquad\left. \exp \br{- \frac{2n \br{{\Delta}/{2} - \bias\br{\II^{(n)}_{\mathsf{EMI}} (\mathbf{X}' \wedge \mathbf{Y}')}}^2}{36 \log^2 n}}\right\}.
            \end{align*}
        \end{lemma}
        \begin{proof}
            Recalling \Cref{eq:delta-def}, we have 
            \begin{align}
                &P\br{\II^{(n)}_{\mathsf{EMI}} (\mathbf{X} \wedge \mathbf{Y}) \geq \II^{(n)}_{\mathsf{EMI}} (\mathbf{X}' \wedge \mathbf{Y}')} \nonumber\\
                &= P\left(\II^{(n)}_{\mathsf{EMI}} (\mathbf{X} \wedge \mathbf{Y}) \right.\nonumber\\
                &\left.\qquad\qquad- \II(X \wedge Y) - \II^{(n)}_{\mathsf{EMI}} (\mathbf{X}' \wedge \mathbf{Y}') + \II(X' \wedge Y') \geq \Delta\right)\nonumber\\
                &\leq P\br{\II^{(n)}_{\mathsf{EMI}} (\mathbf{X} \wedge \mathbf{Y}) - \II(X \wedge Y) \geq \Delta/2} \nonumber\\
                &\qquad\quad+ P\br{\II^{(n)}_{\mathsf{EMI}} (\mathbf{X}' \wedge \mathbf{Y}') - \II(X' \wedge Y') \leq - \Delta/2} \label{eq:estim-order-claim}
            \end{align}
            Using \Cref{lem:emp-mi-conc}, and in view of \Cref{eq:delta-def}, \Cref{eq:lower-bound-on-n},
            \begin{align}
                &P\br{\II^{(n)}_{\mathsf{EMI}} (\mathbf{X} \wedge \mathbf{Y}) - \II(X \wedge Y) \geq \Delta/2}\nonumber\\ 
                &= P\Big(\II^{(n)}_{\mathsf{EMI}} (\mathbf{X} \wedge \mathbf{Y}) -\mathbb{E}\bs{\II^{(n)}_{\mathsf{EMI}} (\mathbf{X} \wedge \mathbf{Y})} \nonumber \\
                &\qquad\qquad\qquad\qquad\qquad\geq \Delta/2 - \bias\br{\II^{(n)}_{\mathsf{EMI}} (\mathbf{X} \wedge \mathbf{Y})}\Big) \nonumber\\
                &\leq \exp \br{- \frac{2n \br{{\Delta}/{2} - \bias\br{\II^{(n)}_{\mathsf{EMI}} (\mathbf{X} \wedge \mathbf{Y})}}^2}{36 \log^2 n}}, \label{eq:5.5}
            \end{align}
            and similarly,
            \begin{align}
                &P\br{\II(X' \wedge Y') - \II^{(n)}_{\mathsf{EMI}} (\mathbf{X}' \wedge \mathbf{Y}') \geq \Delta/2} \nonumber \\
                &\leq \exp \br{- \frac{2n \br{{\Delta}/{2} - \bias\br{\II^{(n)}_{\mathsf{EMI}} (\mathbf{X}' \wedge \mathbf{Y}')}}^2}{36 \log^2 n}}. \label{eq:5.6}
            \end{align}
            The claimed bound follows by using \Cref{eq:5.5,eq:5.6} in \Cref{eq:estim-order-claim}.
        \end{proof}
    \subsection{Bandit algorithm for estimating $\SI(X_\mathcal{M})$}
        The following bandit-based method identifies the best arm pair corresponding to the edge of the MCT across which mutual information is minimal. For an introduction to the fundamentals of bandit algorithms, see \cite{szepesvari-bandits}.   

        In the parlance of banditry, the environment has $m$ arms, one arm corresponding to each vertex in $\mathcal{G} = (\mathcal{M}, \mathcal{E})$. The agent can pull, in any step, two arms that are connected by an edge in $\mathcal{E}$. Each action of the agent is specified by the pair $(i,j)$, $1 \leq i < j \leq m$, $(i,j) \in \mathcal{E}$, with associated reward being the realizations $(X_i = x_i, X_j = x_j)$. The agent is allowed to pull a total of $N$ pairs of arms, say, using \emph{uniform sampling}, where $N$ will be specified below. A pulling of a pair of arms can be viewed also as pulling the corresponding connecting edge, thereby rendering it a traditional stochastic bandit problem. We resort to a uniform sampling strategy for the sake of simplicity.

        \begin{definition}[Uniform sampling]\label{def:unif-samp}
            In uniform sampling, pairs of rvs corresponding to edges of the tree are sampled equally often. Specifically, each pair of rvs $(X_i, X_j)$, $(i,j) \in \mathcal{E}$, is sampled $n$ times over nonoverlapping time instants. Hence, an agent pulls a total of $N$ pairs of arms, where $N = \abs{\mathcal{E}} n$.
        \end{definition}
        
        By means of these actions, the agent seeks to form estimates of all two-dimensional marginal pmfs $P_{X_i X_j}$ and of the corresponding $\II(X_i \wedge X_j)$ for $(i,j)$ as above, and subsequently identify $(\bar{i}, \bar{j}) \in \mathcal{E}$ (see \Cref{rem:min-edge}). Let $X_{\mathcal{M}}^N$ denote $N$ i.i.d. repetitions of $X_\mathcal{M} = (X_1, \ldots, X_m)$. Specifically, the agent must produce an estimate $\hat{e}_N = \hat{e}_N (X_\mathcal{M}^N) \in \mathcal{E}$ of $(\bar{i}, \bar{j}) \in \mathcal{E}$ at the conclusion of $N$ steps so as to minimize the error probability $P(\hat{e}_N \neq (\bar{i}, \bar{j}))$. The following notation is used. Write 
        \begin{align*}
            \II(i \wedge j) = \II(X_i \wedge X_j), \quad (i,j) \in \mathcal{E}
        \end{align*}
        for simplicity, and let 
        \begin{align*}
            {\II}_{\mathsf{EMI}}^{(n)}(i \wedge j) \triangleq {\II}_{\mathsf{EMI}}^{(n)}(\mathbf{X}_i \wedge \mathbf{X}_j)
        \end{align*}
        be the estimate of $\II(i \wedge j)$. At the end of $N = \abs{\mathcal{E}} n$ steps, set 
        \begin{align}
            \hat{e}(X_\mathcal{M}^N) = \arg \min_{(i,j) \in \mathcal{E}} {\II}_{\mathsf{EMI}}^{(n)} (i,j)  = (i^*, j^*), \text{ say} \label{eq:edge-estimate}
        \end{align}
        with ties being resolved arbitrarily. Correspondingly, the estimate of shared information is 
        \begin{align}
            \SI_{\mathsf{EMI}}^{(N)}(X_\mathcal{M}^N) \triangleq {\II}_{\mathsf{EMI}}^{(n)}(i^* \wedge j^*). \label{eq:si-estimate}
        \end{align}
        Denote 
        \begin{align*}
            \Delta_{ij} = \II(X_i \wedge X_j) - \II(X_{\bar{i}}, X_{\bar{j}}), \qquad (i,j) \in \mathcal{E}
        \end{align*}
        and 
        \begin{align*}
            \Delta_1 = \min_{\substack{(i,j) \in \mathcal{E} \\ (i,j) \neq (\bar{i}, \bar{j})}} \II(X_i \wedge X_j) - \II(X_{\bar{i}} \wedge X_{\bar{j}}),
        \end{align*}
        where the latter is the difference between the second-lowest and lowest mutual information across edges in $\mathcal{E}$. Note that $\Delta_1 > 0$ by the assumed uniqueness of the minimizing edge $(\bar{i}, \bar{j})$.

        The shared information estimate $\SI_{\mathsf{EMI}}^{(N)}(X_\mathcal{M}^N)$ converges almost surely and in the mean. This is shown in \Cref{th:si-conv} below. To that end, we first provide an upper bound for the probability of arm misidentification with uniform sampling. 

        \begin{proposition}\label{th:estimation-result}
            For uniform sampling, the probability of error in identifying the optimal pair of arms is 
            \begin{align*}
                P \br{\hat{e}_N(X_\mathcal{M}^N) \neq (\bar{i}, \bar{j})} \leq 2 \abs{\mathcal{E}} \exp\br{\frac{-(N/\abs{\mathcal{E}}) \Delta_1^2}{648 \log^2(N/\abs{\mathcal{E}})}}
            \end{align*}
            for all
            \begin{align}
                N > \abs{\mathcal{E}} \max\bc{\frac{\abs{\mathcal{X}}^2 - 1}{2^{\Delta_1/3} - 1}, \frac{\abs{\mathcal{X}} - 1}{2^{\Delta_1/6} - 1}}. \label{eq:prop-estim-res}
            \end{align}
        \end{proposition}

        \begin{proof}
            With $N = \abs{\mathcal{E}}n$, let $x_\mathcal{M}^N$ represent a realization of $X_\mathcal{M}^N$. For each $(i,j) \in \mathcal{E}$, the agent computes the empirical mutual information estimate $\II_{\mathsf{EMI}}^{(n)} (\mathbf{x}_i \wedge \mathbf{x}_j)$ of $\II(X_i \wedge X_j)$. Note that the sampling of arm pairs occurs over nonoverlapping time instants. By \Cref{lem:mut-inf-bias} and \Cref{eq:lower-bound-on-n}
            \begin{align*}
                \abs{\bias(\II_{\mathsf{EMI}}^{(n)} (i \wedge j))} &\leq \frac{\Delta_1}{3} \leq \frac{\Delta_{ij}}{3} < \frac{\Delta_{ij}}{2} \quad \text{for $(i,j) \neq (\bar{i}, \bar{j})$},
            \end{align*}
            for all $N$ as in \Cref{eq:prop-estim-res}. Then, we have
            \begin{align}
                &P\br{\hat{e}_N (X_\mathcal{M}^N) \neq (\bar{i}, \bar{j})} \nonumber\\
                &= P\br{\II^{(n)}_{\mathsf{EMI}} (\bar{i} \wedge \bar{j}) \geq \II^{(n)}_{\mathsf{EMI}} ({i} \wedge {j}) \text{ for some $(i,j) \neq (\bar{i}, \bar{j})$}} \nonumber \\
                &\leq \sum_{(i,j) \neq (\bar{i}, \bar{j})} P\br{\II^{(n)}_{\mathsf{EMI}} (\bar{i} \wedge \bar{j}) \geq \II^{(n)}_{\mathsf{EMI}} ({i} \wedge {j})} \nonumber\\
                &\leq \sum_{(i,j) \neq (\bar{i}, \bar{j})} 2 \exp\br{\frac{-n\Delta_{ij}^2}{648 \log^2 n}}, \qquad \text{by \Cref{lem:estim}}\nonumber\\
                &\leq 2 \abs{\mathcal{E}} \exp\br{\frac{-(N/\abs{\mathcal{E}}) \Delta_1^2}{648 \log^2(N/\abs{\mathcal{E}})}}. \nonumber \qedhere
            \end{align}
        \end{proof}
        \begin{theorem}\label{th:si-conv}
            For uniform sampling (as in \Cref{def:unif-samp}), the estimate $\SI_{\mathsf{EMI}}^{(N)}(X_\mathcal{M}^N)$ converges as $N \rightarrow \infty$ to $\SI(X_\mathcal{M})$ almost surely and in the mean.
        \end{theorem}
        \begin{proof}
            Let $0 < \epsilon < \Delta_1$. By \Cref{eq:lower-bound-on-n}, we can choose $n$ large enough such that $\bias(\II_{\mathsf{EMI}}^{(n)}(\bar{i} \wedge \bar{j})) < \epsilon/2$ (see \Cref{rem:min-edge}). For all such $n$,
            \begin{align}
                &P\br{\abs{\SI_{\mathsf{EMI}}^{(N)}(X_\mathcal{M}) - \SI(X_\mathcal{M})} > \epsilon}\nonumber\\
                &\leq P\br{\abs{\SI_{\mathsf{EMI}}^{(N)}(X_\mathcal{M}^N) - \II(\bar{i} \wedge \bar{j})} > \epsilon, \hat{e}_N(X_\mathcal{M}^N) = (\bar{i}, \bar{j})} \nonumber\\
                &\qquad\qquad\qquad\qquad\qquad\qquad\qquad+ P\br{\hat{e}_N(X_\mathcal{M}^N) \neq (\bar{i} \wedge \bar{j})} \nonumber \\
                &\leq P\br{\abs{\II_{\mathsf{EMI}}^{(n)}(\bar{i} \wedge \bar{j}) - \II(\bar{i} \wedge \bar{j})}  > \epsilon} + P\br{\hat{e}_N(X_\mathcal{M}^N) \neq (\bar{i}, \bar{j})}\nonumber\\
                &= P\bigg(\bigg|\II_{\mathsf{EMI}}^{(n)}(\bar{i} \wedge \bar{j}) - \E{\II_{\mathsf{EMI}}^{(n)}(\bar{i} \wedge \bar{j})} \nonumber\\
                &\qquad\qquad\qquad\qquad + \E{\II_{\mathsf{EMI}}^{(n)}(\bar{i} \wedge \bar{j})} - \II(\bar{i} \wedge \bar{j})\bigg|  > \epsilon\bigg) \nonumber \\
                &\qquad\qquad\qquad\qquad\qquad\qquad\qquad+ P\br{\hat{e}_N(X_\mathcal{M}^N) \neq (\bar{i}, \bar{j})}\nonumber\\
                &\leq P\bigg(\abs{\II_{\mathsf{EMI}}^{(n)}(\bar{i} \wedge \bar{j}) - \E{\II_{\mathsf{EMI}}^{(n)}(\bar{i} \wedge \bar{j})}}\nonumber \\
                &\qquad\qquad\qquad\qquad+ \abs{\bias\br{\II_{\mathsf{EMI}}^{(n)}(\bar{i} \wedge \bar{j})}}  > \epsilon\bigg) \nonumber\\
                &\qquad\qquad\qquad\qquad\qquad\qquad\qquad+ P\br{\hat{e}_N(X_\mathcal{M}^N) \neq (\bar{i}, \bar{j})}\nonumber\\
                &\leq \exp\br{\frac{-2(N/\abs{\mathcal{E}})(\epsilon/2)^2}{36 \log^2 (N/\abs{\mathcal{E}})}} \nonumber\\
                &\qquad\qquad\qquad\qquad+ 2 \abs{\mathcal{E}} \exp\br{\frac{-(N/\abs{\mathcal{E}}) \Delta_1^2}{648 \log^2(N/\abs{\mathcal{E}})}} \label{eq:th-4.2}
            \end{align}
            by \Cref{lem:emp-mi-conc} and \Cref{th:estimation-result}. Almost sure convergence follows from \Cref{eq:th-4.2} and the Borel-Cantelli Lemma (cf. e.g., \cite[Lemma 7.3]{koralov2007theory}); and furthermore, since $\SI_{\mathsf{EMI}}^{(N)}(X_\mathcal{M}^N) \leq \log \abs{\mathcal{X}}$ for all $N$, almost sure convergence implies convergence in the mean by the Dominated Convergence Theorem (cf. e.g., \cite[Theorem 3.27]{koralov2007theory}).
        \end{proof}
        The following corollary specifies the sample complexity of $\SI^{(N)}_{\mathsf{EMI}}$ in terms of a minimum requirement on $N$ for which the estimation error is small with high probability. 
        \begin{corollary}\label{cor:si-est-sample-complexity}
            For $0 < \epsilon < 1/2$ and $\delta < 1/e$, we have 
            \begin{align*}
                P\br{\abs{\SI_{\mathsf{EMI}}^{(N)}(X_\mathcal{M}^N) - \SI(X_\mathcal{M})} > \epsilon} \leq \delta
            \end{align*}
            for sample complexity $N = N(\epsilon, \delta)$ that obeys\footnote{The approximate form of \Cref{eq:sample-complex} considers only the significant terms depending on $\abs{\mathcal{X}}$, $\abs{\mathcal{E}}$, $\Delta_1$, $\epsilon$ and $\delta$.}
            \begin{align}
                N &\gtrsim \abs{\mathcal{E}} \left[\frac{\abs{\mathcal{X}}}{\epsilon} + \frac{1}{\epsilon^2} \ln\br{\frac{1}{\delta}} \log^2 \br{\frac{1}{\epsilon^2} \ln\br{\frac{1}{\delta}}} \right.\nonumber\\
                &\qquad\left. + \frac{1}{\Delta_1^2} \ln\br{\frac{\abs{\mathcal{E}}}{\delta}} \log \br{\frac{1}{\Delta_1^2} \ln\br{\frac{\abs{\mathcal{E}}}{\delta}}} \right.\nonumber\\
                &\qquad\left. + \frac{1}{\Delta_1^2} \ln\br{\frac{\abs{\mathcal{E}}}{\delta}} \log^2 \br{\frac{1}{\Delta_1^2} \ln\br{\frac{\abs{\mathcal{E}}}{\delta}}}\right]. \label{eq:sample-complex}
            \end{align}
        \end{corollary}
        The proof of the corollary relies on the following technical lemma, which is similar in spirit to \cite[Lemma A.1]{sss-ml-book}.
        \begin{lemma}\label{lem:tech-1}
            It holds that 
            \begin{align*}
                x \geq c \ln^2 x, \quad c \geq 1,\ x \geq \max\bc{1, 4c \ln 2c + 16c \ln^2 c}.
            \end{align*}
        \end{lemma}
        \begin{proof}
            See Appendix C.
        \end{proof}
        \begin{proof}[Proof of \Cref{cor:si-est-sample-complexity}]
            From \Cref{eq:th-4.2},
            \begin{align*}
                P\br{\abs{\SI_{\mathsf{EMI}}^{(N)}(X_\mathcal{M}^N) - \SI(X_\mathcal{M})} > \epsilon} \leq \delta,
            \end{align*}
            for $n = N/\abs{\mathcal{E}}$ satisfying
            \begin{align}
                n \geq \frac{\abs{\mathcal{X}}}{\epsilon},\ \frac{n}{\log^2 n} \geq \frac{1}{\epsilon^2} \ln\br{\frac{1}{\delta}},\ \frac{n}{\log^2 n} \geq \frac{1}{\Delta_1^2} \ln\br{\frac{\abs{\mathcal{E}}}{\delta}} \label{eq:samp-comp-proof}
            \end{align}
            up to numerical constant factors. Each of the inequalities in \Cref{eq:samp-comp-proof} yields one or more lower bounds for $n = N/\abs{\mathcal{E}}$; the first does so directly, and the latter two upon writing them as $n \geq c \log^2 n$ (where $c$ does not depend on $n$) and using \Cref{lem:tech-1}. The conditions on $\epsilon$ and $\delta$ in \Cref{cor:si-est-sample-complexity}, allow us to drop one of the bounds since it is always weaker than another. Combining all the lower bounds obtained from \Cref{eq:samp-comp-proof} and using $N = \abs{\mathcal{E}} n$ finally results in \Cref{eq:sample-complex}.
        \end{proof}

\section{Closing Remarks}\label{sec:concluding}
        While the Hammersley-Clifford Theorem \cite[Theorem 3.9]{lauritzen1996} can be used to show the equivalence in \Cref{th:g-g} between the MCT definition and the global Markov property, when the joint pmf of $X_\mathcal{M}$ is strictly positive, we show for an MCT that it holds even for pmfs that are \emph{not} strictly positive. The tree structure plays a material role in our proof. In particular, agglomeration of connected subsets of an MCT form an MCT as in \Cref{def:agglomerated-tree} and \Cref{lem:agglom-mct}. The MCT property only involves verifying the Markov condition \Cref{eq:mct} or \Cref{eq:mct-mc-def} for each edge in the tree, and therefore is easier to check than the global Markov property \Cref{eq:g-g}. Joint pmfs that are not strictly positive arise in applications such as function computation when a subset of rvs are determined by another subset. 
    
        \Cref{thm:si-mct} shows that for an MCT, a simple 2-partition achieves the minimum in \Cref{def:si}. While the result in \Cref{thm:si-mct} was known \cite{csiszar-narayan-secrecy-capacities}, our proof uses new techniques and further implies that for \emph{any} partition $\pi$ with disconnected atoms, there is a partition with connected atoms that has $\cII$-value (see \Cref{eq:sub-si-eq}) less than or equal to that of $\pi$. This structural property is stronger than that needed for proving \Cref{thm:si-mct}. 
    
        Our proof technique for \Cref{thm:si-mct} can serve as a stepping stone for analyzing SI for more complicated graphical models in which the underlying graph is not a tree; see \cite{sb-pn-isit-2023}. In particular, the tree structure was used in the proof of \Cref{thm:si-mct} only in Step 1 and \Cref{eq:6} in Step 2.
    
        In \Cref{sec:estimating-si}, we have presented an algorithm for best-arm identification with biased (and asymptotically unbiased) estimates. A uniform sampling strategy and the empirical mutual information estimator were chosen for simplicity. Using bias-corrected estimators like the Miller-Madow or jack-knifed estimator for mutual information \cite{paninski} would improve the bias performance of the algorithm. However, it hurts the constant in the bounded differences inequality that appears in \Cref{lem:emp-mi-conc}. Polynomial approximation-based estimators \cite{jiao-minimax} could also improve sample complexity. Moreover, a successive rejects algorithm \cite{audibert-best-arm-identification}, \cite{boda-prashanth-correlated-bandits} could yield a better sample complexity than uniform sampling for a fixed estimator, as hinted by \cite{audibert-best-arm-identification}, \cite{boda-prashanth-correlated-bandits} in different settings. The precise tradeoff afforded by the choice of a better estimator remains to be understood, as does the sample complexity of more refined algorithms for best arm identification in our setting. Both demand a converse result that needs to take into account estimator bias; this remains under study in our current work. A converse would also settle the question of optimality of the $\exp(-O(N/\log^2 N))$ decay in the probability of error in \Cref{th:estimation-result}.

\appendices

\section{Proof of \Cref{lem:mrf-mi}}\label{app:1}
    \noindent We have
    \begin{align}
        &\II(X_{\mathcal{B}(i \leftarrow j)} \wedge X_{\mathcal{B}(j \leftarrow i)}) \nonumber \\
        &= \II(X_i \wedge X_j) + \II(X_i \wedge X_{\mathcal{B}(j \leftarrow i) \setminus \bc{j}} \cond X_j)\nonumber \\
        &\qquad\qquad+ \II(X_{\mathcal{B}(i \leftarrow j) \setminus \bc{i}} \wedge X_j \cond X_i)\nonumber\\
        &\qquad\qquad\qquad+ \II(X_{\mathcal{B}(i \leftarrow j) \setminus \bc{i}} \wedge X_{\mathcal{B}(j \leftarrow i) \setminus \bc{j}} \cond X_i, X_j) \nonumber \\ 
        &= \II(X_i \wedge X_j) + \II(X_{\mathcal{B}(i \leftarrow j) \setminus \bc{i}} \wedge X_{\mathcal{B}(j \leftarrow i) \setminus \bc{j}} \cond X_i, X_j)\nonumber\\
        &= \II(X_i \wedge X_j) + \HH(X_{\mathcal{B}(i \leftarrow j) \setminus \bc{i}} \cond X_i)\nonumber \\
        &\qquad\qquad\qquad\qquad- \HH(X_{\mathcal{B}(i \leftarrow j) \setminus \bc{i}} \cond X_i, X_{\mathcal{B}(j \leftarrow i)})\label{eq:mct-proof}
    \end{align}
    where the previous two inequalities are by \Cref{eq:mct-mc-def}.

    The claim of \Cref{lem:mrf-mi} would follow from \Cref{eq:mct-proof} upon showing that 
    \begin{align}
        \HH(X_{\mathcal{B}(i \leftarrow j) \setminus \bc{i}} \cond X_i, X_{\mathcal{B}(j \leftarrow i)}) = \HH(X_{\mathcal{B}(i \leftarrow j) \setminus \bc{i}} \cond X_i) \label{eq:mct-mi-proof-claim}. 
    \end{align}
    
    Without loss of generality, set $j$ to be the root of the tree; this defines a \emph{directed} tree whose leaves are from among the vertices (in $\mathcal{M}$) with no descendants. Denote the parent of $i'$ in the (directed) tree by $\parent(i')$. Note that $\parent(i) = j$ in \Cref{eq:mct-mi-proof-claim}. We shall use induction on the \emph{height} of $i'$, i.e., the maximum distance of $i'$ from a leaf of the directed tree, to show that 
    \begin{align}
        &\HH(X_{\mathcal{B}(i' \leftarrow \parent(i')) \setminus \bc{i'}} \cond X_{i'}, X_{\mathcal{B}(\parent(i') \leftarrow i')}) \nonumber\\
        &\qquad\qquad\qquad\qquad\qquad= \HH(X_{\mathcal{B}(i' \leftarrow \parent(i')) \setminus \{i'\}} \cond X_{i'}) \label{eq:mct-mi-proof-claim-2}, 
    \end{align}
    which proves \Cref{eq:mct-mi-proof-claim} upon setting $i'=i$ and $\parent(i') = \parent(i) = j$. 
    
    First, assume that $i'$ is a leaf. Then $\mathcal{B}(i' \leftarrow \parent(i')) \setminus \bc{i'} = \varnothing$ and \Cref{eq:mct-mi-proof-claim-2} holds trivially. 
    \begin{figure}[htb]
        \centering
        \includegraphics[width=0.48\textwidth]{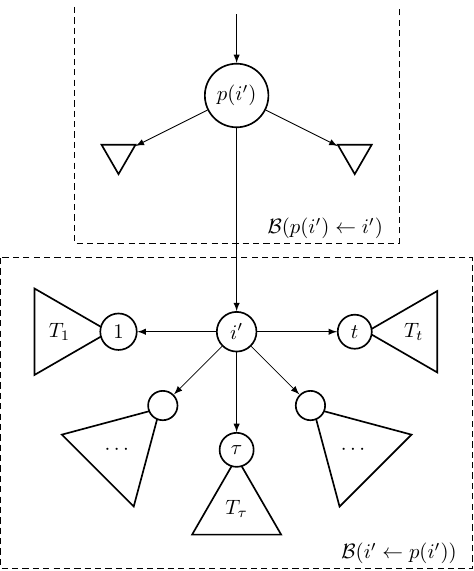}
        \caption{Schematic for proof of \Cref{eq:mct-mi-proof-claim-2}.}
        \label{fig:diag-proof-in-app-a}
    \end{figure}
    
    Next, assume the induction hypothesis that \Cref{eq:mct-mi-proof-claim-2} is true for all vertices at height $< h$, and consider a vertex $i'$ at height $h$. Let $i'$ have children $1, \ldots, t$; each of these vertices is the root of subtree $T_\tau = \mathcal{B}(\tau \leftarrow i')$, $1 \leq \tau \leq t$. See \Cref{fig:diag-proof-in-app-a}. Further, each vertex $\tau$, $1 \leq \tau \leq t$, has height $< h$. Then in \Cref{eq:mct-mi-proof-claim-2},
    \begin{align}
        &\HH(X_{\mathcal{B}(i' \leftarrow \parent(i')) \setminus \bc{i'}} \cond X_{i'}, X_{\mathcal{B}(\parent(i') \leftarrow i')})\nonumber\\
        &= \HH\br{(X_{T_\tau \setminus \bc{\tau}}, X_\tau)_{1 \leq \tau \leq t} \cond X_{i'}, X_{\mathcal{B}(\parent(i') \leftarrow i')}}\nonumber\\
        &= \sum_{\tau = 1}^t \left[\HH \br{X_\tau \cond \br{X_{T_\sigma}}_{1 \leq \sigma \leq \tau - 1}, X_{i'}, X_{\mathcal{B}(\parent(i') \leftarrow i')}}\right.\nonumber\\
        &\left. + \HH\br{X_{T_\tau \setminus \bc{\tau}} \cond X_\tau, \br{X_{T_\sigma}}_{1 \leq \sigma \leq \tau - 1}, X_{i'}, X_{\mathcal{B}(\parent(i') \leftarrow i')}}\right].\label{eq:3.3}
    \end{align}
    In \Cref{eq:3.3}, for each $\tau$, $1 \leq \tau \leq t$, the first term within $\bs{\cdot}$ is 
    \begin{align}
        &\HH \br{X_\tau \cond \br{X_{T_\sigma}}_{1 \leq \sigma \leq \tau - 1}, X_{i'}, X_{\mathcal{B}(\parent(i') \leftarrow i')}}\nonumber\\
        &\qquad = \HH(X_\tau \cond X_{i'}) = \HH \br{X_\tau \cond \br{X_{T_\sigma}}_{1 \leq \sigma \leq \tau - 1}, X_{i'}} \label{eq:app-1-proof-1}
    \end{align}
    by \Cref{eq:mct-mc-def} since
    \begin{align*}
        \br{\bigcup_{\sigma = 1}^{\tau - 1} T_{\sigma}, \mathcal{B}(\parent(i') \leftarrow i')} \subseteq \mathcal{B}(i' \leftarrow \tau)
    \end{align*}
    (see \Cref{fig:diag-proof-in-app-a}). In the second term in $\bs{\cdot}$, we apply the induction hypothesis to vertex $\tau$ which is at height $h-1$. Note that $\parent(\tau) = i'$. Since
    \begin{align*}
        &X_{T_\tau \setminus \bc{\tau}} = X_{\mathcal{B}(\tau \leftarrow \parent(\tau)) \setminus \bc{\tau}}\\
        &\quad\text{and} \br{\br{X_{T_\sigma}}_{1 \leq \sigma \leq \tau - 1}, X_{i'}, X_{\mathcal{B}(\parent(i') \leftarrow i')}} \subseteq \mathcal{B}(\parent(\tau) \leftarrow \tau),
    \end{align*}
    by the induction hypothesis at vertex $\tau$, we get
    \begin{align}
        &\HH\br{X_{T_\tau \setminus \bc{\tau}} \cond X_\tau, \br{X_{T_\sigma}}_{1 \leq \sigma \leq \tau - 1}, X_{i'}, X_{\mathcal{B}(\parent(i') \leftarrow i')}}\nonumber\\
        &= \HH\br{X_{T_\tau \setminus \bc{\tau}} \cond X_\tau} \nonumber\\
        &= \HH\br{X_{T_\tau \setminus \bc{\tau}} \cond X_\tau, \br{X_{T_\sigma}}_{1 \leq \sigma \leq \tau - 1}, X_{i'}} \label{eq:app-1-proof-2}
    \end{align}
    with the last equality being due to \Cref{eq:mct-mc-def}. Substituting \Cref{eq:app-1-proof-1}, \Cref{eq:app-1-proof-2} in \Cref{eq:3.3}, we obtain
    \begin{align*}
        &\HH(X_{\mathcal{B}(i' \leftarrow \parent(i')) \setminus \bc{i'}} \cond X_{i'}, X_{\mathcal{B}(\parent(i') \leftarrow i')})\\
        &= \sum_{\tau = 1}^t \left[\HH \br{X_\tau \cond \br{X_{T_\sigma}}_{1 \leq \sigma \leq \tau - 1}, X_{i'}}\right.\\
        &\left.\qquad\qquad\qquad + \HH\br{X_{T_\tau \setminus \bc{\tau}} \cond X_\tau, \br{X_{T_\sigma}}_{1 \leq \sigma \leq \tau - 1}, X_{i'}}\right]\\
        &= \sum_{\tau=1}^t \HH\br{X_{T_\tau} \cond \br{X_{T_\sigma}}_{1 \leq \sigma \leq \tau - 1}, X_{i'}}\\
        &= \HH(X_{\mathcal{B}(i' \leftarrow \parent(i')) \setminus \{i'\}} \cond X_{i'})
    \end{align*}
    (see \Cref{fig:diag-proof-in-app-a}) which is \Cref{eq:mct-mi-proof-claim-2}. \hfill \qed

\section{Proof of \Cref{lem:mct-local-prop} and \Cref{th:g-g}}\label{app:2}
\begin{proof}[Proof of \Cref{lem:mct-local-prop}]
    Considering first \Cref{eq:mct-local-prop-1}, suppose that vertex $i \in \mathcal{M}$ has $k$ neighbor, with $\mathcal{N}(i) = \bc{i_1, \ldots, i_k}$, $1 \leq k \leq m-1$. Then
    \begin{align*}
        \mathcal{M} \setminus (\bc{i} \cup \mathcal{N}(i)) = \bigcup_{l=1}^k \mathcal{B}(i_l \leftarrow i) \setminus \bc{i_l}.
    \end{align*}
    The claim of the lemma is 
    \begin{align}
        X_i \mc \br{X_{i_u}}_{1 \leq u \leq k} \mc \br{X_{\mathcal{B}(i_l \leftarrow i) \setminus \bc{i_l}}}_{1 \leq l \leq k}. \label{eq:mct-local-prop-proof-1}
    \end{align}
    We have 
    \begin{align}
        &\II\br{X_i \wedge \br{X_{\mathcal{B}(i_l \leftarrow i) \setminus \bc{i_l}}}_{1 \leq l \leq k} \bigcond \br{X_{i_u}}_{1 \leq u \leq k}} \nonumber\\
        &= \sum_{l=1}^k \II\left(X_i \wedge X_{\mathcal{B}(i_l \leftarrow i) \setminus \bc{i_l}} \bigcond \br{X_{\mathcal{B}(i_j \leftarrow i) \setminus \bc{i_j}}}_{1 \leq j \leq l-1},\right. \nonumber \\ 
        &\qquad\qquad\qquad\qquad\qquad\qquad\qquad\qquad\qquad\left.\br{X_{i_u}}_{1 \leq u \leq k}\right)\nonumber\\
        &\leq \sum_{l=1}^k \II\Big(\bs{X_i, \br{X_{\mathcal{B}(i_j \leftarrow i) \setminus \bc{i_j}}}_{1 \leq j \leq l-1}, \br{X_{i_u}}_{1 \leq u \neq l \leq k}} \nonumber \\ 
        &\qquad\qquad\qquad\qquad\qquad\qquad\wedge X_{\mathcal{B}(i_l \leftarrow i) \setminus \bc{i_l}} \bigcond X_{i_l}\Big). \label{eq:mct-local-prop-proof-2}
    \end{align}
    For each $l$, $1 \leq l \leq k$, the rvs within $\bs{\cdot}$ above have indices that lie in $\mathcal{B}(i \leftarrow i_l) \setminus \bc{i_l}$. Hence, by \Cref{lem:mrf-mi} (specifically \Cref{eq:mrf-mi-2b}), each term in the sum in \Cref{eq:mct-local-prop-proof-2} equals zero. This proves \Cref{eq:mct-local-prop-proof-1}. See \Cref{diag:local_prop_diag}.

    \begin{figure}[htbp]
        \centering
        \includegraphics[width=0.48\textwidth]{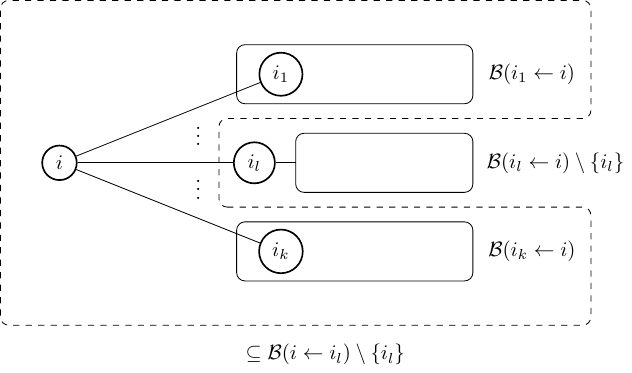}
        \caption{Schematic for the proof of \Cref{lem:mct-local-prop}.}
        \label{diag:local_prop_diag}
    \end{figure}

    Turning to \Cref{eq:mct-local-prop-2}, we have
    \begin{align*}
        &\II\br{X_A \wedge X_{\mathcal{M} \setminus (A \cup \mathcal{N}(A))} \cond \mathcal{N}(A)} \\
        &= \II\br{(X_i, i \in A) \wedge X_{\mathcal{M} \setminus \bigcup_{u \in A} \br{\bc{u} \cup \mathcal{N}(u)}} \bigcond X_{\bigcup_{v \in A} \mathcal{N}(v)}}\\
        &\leq \sum_{i \in A} \II\left(X_i \wedge X_{\bigcup_{j \in A \setminus \bc{i}} (\bc{j} \cup \mathcal{N}(j))},\right.\\  
        &\qquad\qquad\qquad\qquad\qquad\qquad\left.X_{\mathcal{M} \setminus \bigcup_{u \in A} \br{\bc{u} \cup \mathcal{N}(u)}} \bigcond X_{\mathcal{N}(i)}\right)\\
        &= \sum_{i \in A} \II\left(X_i \wedge X_{\br{\bigcup_{j \in A \setminus \bc{i}} (\bc{j} \cup \mathcal{N}(j))}\setminus \mathcal{N}(i)},\right.\\ 
        &\qquad\qquad\qquad\qquad\qquad\qquad\left.X_{\mathcal{M} \setminus \bigcup_{u \in A} \br{\bc{u} \cup \mathcal{N}(u)}} \bigcond X_{\mathcal{N}(i)}\right)\\
        &= 0
    \end{align*}
    by \Cref{eq:mct-local-prop-1} since for each $i \in A$,
    \begin{align*}
        &\br{\br{\bigcup_{j \in A \setminus \bc{i}} (\bc{j} \cup \mathcal{N}(j))}\setminus \mathcal{N}(i)} \\
        &\qquad\cup \br{\mathcal{M} \setminus \bigcup_{u \in A} \br{\bc{u} \cup \mathcal{N}(u)}} \subseteq \mathcal{M} \setminus (\bc{i} \cup \mathcal{N}(i)). \qedhere
    \end{align*}
\end{proof}
\begin{proof}[Proof of \Cref{th:g-g}]
    The converse claim is immediately true upon choosing: for every $(i,j) \in \mathcal{E}$, $A = \mathcal{B}(i \leftarrow j) \setminus \bc{i}$, $S = {i}$, $B = \bc{j}$.

    Turning to the first claim, let 
    \begin{align*}
        A = \bigsqcup_{\alpha = 1}^a A_\alpha, \quad B = \bigsqcup_{\beta = 1}^b B_\beta, \quad S = \bigsqcup_{\sigma = 1}^s S_\sigma
    \end{align*}
    be representations in terms of maximally connected subsets of $A$, $B$ and $S$, respectively. With $N = \mathcal{M} \setminus (A \cup B \cup S)$, let $N = \sqcup_{\nu = 1}^n N_\nu$ be a decomposition into maximally connected subsets of $N$. Denote 
    \begin{align*}
        \mathcal{A} = \bc{A_\alpha, 1 \leq \alpha \leq a}, \quad \mathcal{B} = \bc{B_\beta, 1 \leq \beta \leq b},\\
        \quad \mathcal{S} = \bc{S_\sigma, 1 \leq \sigma \leq s}, \quad \mathcal{N} = \bc{N_\nu, 1 \leq \nu \leq n}.
    \end{align*}
    Referring to \Cref{def:agglomerated-tree} and recalling \Cref{lem:agglom-mct}, the tree $\mathcal{G}' = (\mathcal{M}', \mathcal{E}')$ with vertex set $\mathcal{M}' = \mathcal{A} \cup \mathcal{B} \cup \mathcal{S} \cup \mathcal{N}$ and edge set in the manner of \Cref{def:agglomerated-tree} constitutes an agglomerated MCT. 

    Next, we observe that since each $N_\nu \in \mathcal{N}$, $1 \leq \nu \leq n$, is maximally connected in N, the neighbors of $N_\nu$ in $\mathcal{G}'$ cannot be in $\mathcal{N}$. Therefore, neighbors of a given $N_\nu$ in $\mathcal{G}'$ that are not in $\mathcal{S}$ must be in $\mathcal{A}$ or $\mathcal{B}$. However, $N_\nu$ cannot have a non$\mathcal{S}$ neighbor in $\mathcal{A}$ and also one in $\mathcal{B}$, for then $A$ and $B$ would not be separated by $S$ in $\mathcal{G}$. Accordingly, for \emph{each} $N_\nu$ in $\mathcal{N}$, if its non$\mathcal{S}$ neighbors in $\mathcal{G}'$ are only in $\mathcal{A}$, add $N_\nu$ to $\mathcal{A}$; let $N'$ be the union of all such $N_\nu$s. Consider $A' = A \cup N'$ and write $A' = \sqcup_{\alpha=1}^{a'} A_\alpha'$ where the $A_\alpha'$s are maximally connected subsets of $A'$. Let $\mathcal{A}' = \bc{A_\alpha', 1\leq \alpha \leq a'}$.

    Now note that $\mathcal{A}'$ and $\mathcal{B}$ are separated in $\mathcal{G}'$ by $\mathcal{S}$. Thus, to establish \Cref{eq:g-g}, it suffices to show the (stronger) assertion 
    \begin{align}
        X_{\mathcal{A}'} \mc X_\mathcal{S} \mc X_\mathcal{B}. \label{eq:global-proof-1}
    \end{align}
    By the description of $\mathcal{A}'$, each of its components (maximal subsets of $A'$) has its neighborhood in $\mathcal{G}'$ that is contained \emph{fully} in $\mathcal{S}$. Let $\tilde{\mathcal{S}} \subseteq \mathcal{S}$ denote the union of all such neighborhoods. Then, by \Cref{lem:mct-local-prop} (\Cref{eq:mct-local-prop-2}) applied to the agglomerated tree $\mathcal{G}'$, since there is no edge in $\mathcal{G}'$ that connects any two elements of $\mathcal{A}'$,
    \begin{align*}
        X_{\mathcal{A}'} \mc X_{\tilde{\mathcal{S}}} \mc X_{\mathcal{M}' \setminus (\mathcal{A}' \cup \tilde{\mathcal{S}})}
    \end{align*}
    so that
    \begin{align}
        0 &= \II(X_{\mathcal{A}'} \wedge X_{\mathcal{M}' \setminus (\mathcal{A}' \cup \tilde{\mathcal{S}})} \cond X_{\tilde{\mathcal{S}}}) \nonumber\\
        &= \II(X_{\mathcal{A}'} \wedge X_{\mathcal{M}' \setminus (\mathcal{A}' \cup \tilde{\mathcal{S}})}, X_{\mathcal{S}\setminus\tilde{\mathcal{S}}} \cond X_{\tilde{\mathcal{S}}}) \nonumber\\
        &\geq \II(X_{\mathcal{A}'} \wedge X_{\mathcal{M}' \setminus (\mathcal{A}' \cup \mathcal{S})} \cond X_\mathcal{S}) \label{eq:global-proof-2}
    \end{align}
    since $\tilde{\mathcal{S}} \subseteq \mathcal{S}$. Finally, \Cref{eq:global-proof-2} implies \Cref{eq:global-proof-1} as $\mathcal{B} \subseteq \mathcal{M}'\setminus(\mathcal{A}' \cup \mathcal{S})$.
\end{proof}
\section{Proof of \Cref{lem:tech-1}} \label{app:3}
\begin{proof}
    For $1 \leq c \leq 1.2$, $x \geq c \ln^2 x$ holds unconditionally; so assume that $c \geq 1.2$. Consider the function $f(x) = x - c \ln^2 x$. Then, using \cite[Lemma A.1]{sss-ml-book}, $x \geq 4c \ln 2c$ implies $x \geq 2c \ln x$ which, in turn, implies $f'(x) \geq 0$. Therefore, for $x \geq 4c \ln c$, $f(x)$ is increasing in $x$. It is easy to check numerically that $f(16c \ln^2 c)$ is positive for $c \geq 1.2$. Thus, for all $x \geq \max\{4c \ln 2c, 16c \ln^2 c\}$, $f(x) \geq 0$ and so $x \geq c \ln^2 x$. 
\end{proof}



\ifCLASSOPTIONcaptionsoff
  \newpage
\fi

\begin{IEEEbiographynophoto}{Sagnik Bhattacharya}
received the Bachelor of Technology in Electrical Engineering from the Indian Institute of Technology Kanpur, India, in 2019. He is currently a PhD candidate in the Department of Electrical and Computer Engineering at the University of Maryland, College Park. His research interests are in information theory, statistical learning, and their practical applications.
\end{IEEEbiographynophoto}


\begin{IEEEbiographynophoto}{Prakash Narayan}
  received the Bachelor of Technology degree in Electrical Engineering from the
  Indian Institute of Technology, Madras in $1976$.
  He received the M.S. degree in Systems Science and Mathematics in $1978$ and
  the D.Sc. degree in Electrical Engineering in $1981$, both from Washington
  University, St. Louis, MO.
  
  He is Professor of Electrical and Computer Engineering at the University
  of Maryland, College Park, with a joint appointment at the Institute for
  Systems Research. His research interests are in network information theory,
  coding theory, communication theory, communication networks, statistical learning, 
  and cryptography.
\end{IEEEbiographynophoto}




\end{document}